\documentclass[12pt]{article}

\usepackage{amsfonts,amsmath, amssymb}
\usepackage{dutchcal}
\usepackage{hyperref}

\usepackage{xcolor}

\usepackage{enumerate}
\usepackage{epsf,epsfig,amsfonts,a4wide}
\usepackage{amsfonts, mathdots}
\usepackage{amsmath,amssymb,amsthm}
\usepackage{url}
\usepackage{amsmath,amssymb,amsthm}

\usepackage{amssymb}
\usepackage{amsthm}
\usepackage{epsf,epsfig,amsfonts,graphicx}

\newtheorem{Theorem}{Theorem}

\newtheorem{Lemma}{Lemma}[section]

\newtheorem{Comment}{Remark}
\newtheorem{Fact}{Fact}

\newtheorem{Ex}{Example}[section]

\theoremstyle{remark}

\newcommand{\be}{\begin{equation}}
\newcommand{\ee}{\end{equation}}

\newcommand{\Id}{\textrm{\rm Id}}

\newcommand{\ddd}{\mathrm{d}}

\newcommand{\pd}[2]{\frac{\partial#1}{\partial#2}}

\newcommand{\tr}{\operatorname{tr}}

\newcommand{\const}{\mathrm{const}}

\newcommand{\weg}[1]{}

\usepackage{csquotes}
\usepackage[isbn=false]{biblatex}  
\title{ Lax pairs for BKM hierarchy}
\author{  
Andrey Yu.\  Konyaev\footnote{Faculty of Mechanics and Mathematics and Center for Fundamental and Applied Mathematics, Moscow State University, 119992, Moscow, Russia  and Institute of Mathematics and Mathematical Modeling, Almaty, Kazakhstan
 \ \ \quad{\tt  maodzund@yandex.ru}}  \, and 
    Vladimir S.\ Matveev\footnote{
Institut f\"ur Mathematik, Friedrich Schiller Universit\"at Jena,
07737 Jena,  Germany  \ \ \quad {\tt  vladimir.matveev@uni-jena.de}}}

\date{ }

\begin{document}

\maketitle

\begin{abstract} 
We construct Lax pairs for the recently (2023) introduced  integrable PDE systems known as the BKM equations. As  many known and previously studied integrable systems are special cases of the BKM systems, our construction provides Lax pairs for many integrable hierarchies, including previously studied ones such as Camassa-Holm,  Dullin-Gottwald-Holm, cKdV, Ito,  and Marvan–Pavlov, as well as new ones. The corresponding pair is related to a Sturm–Liouville operator on the real line whose potential depends rationally on the spectral parameter.
\end{abstract}
\tableofcontents

\section{Introduction}
A \emph{Lax pair} (named after P.\,D.~Lax, who introduced it in \cite{Lax68}) for an evolutional  equation is a pair of linear operators or matrices
\(
\mathbb{L}(t),\; \mathbb{P}(t)
\)
such that the dynamics can be written as the \emph{Lax equation}
\begin{equation}\label{eq:Lax}
\frac{d\mathbb{L}}{dt} \;=\; [\mathbb{P},\mathbb{L}] \;=\; \mathbb{P}\mathbb{L} - \mathbb{L}\mathbb{P}.
\end{equation}
In many interesting PDE/ODE examples, $\mathbb{L}$ is a differential operator depending on the unknown functions  (referred to as field variables), while $\mathbb{P}$ is another operator depending on  the same field variables.

It is known (and easy to prove) that 
if $\mathbb{L}$ evolves by \eqref{eq:Lax}, then all spectral invariants of $\mathbb{L}$ are constant in time. In particular, in the ODE case,
\[
\frac{d}{dt}\,\mathrm{tr}\big(\mathbb{L}^k\big) \;=\; \mathrm{tr}\big([\mathbb{P},\mathbb{L}^k]\big) \;=\; 0, \qquad k=1,2,\dots .
\]
Thus the eigenvalues of 
$\mathbb{L}(t)$ are conserved  (a property sometimes referred to as \emph{isospectrality}), thereby providing integrals of motion. The existence of a (nontrivial) Lax pair is sometimes considered a criterion for the integrability of a dynamical system.  

A canonical PDE example is the famous Korteweg–de Vries (KdV) equation.
For its one field variable $u=u(x,t)$, the equation  
\begin{equation} \label{eq:KdV}
    u_t = \frac{1}{2} u_{xxx}+ \frac{3}{2}  u u_x  
\end{equation}

has Lax pair
\begin{equation} \label{eq:lax_old}
\mathbb{L} =\partial_x^2 + \frac{1}{2}u(x,t),
\qquad
\mathbb{P} = 2\,\partial_x^3 + \frac{3}{2}u(x,t)\,\partial_x + \frac{3}{4}u_x(x,t)
\end{equation}
It is known and easy to check that for these $\mathbb{L}$ and $\mathbb{P}$ 
the equation  $\mathbb{L}_t=[\mathbb{P},\mathbb{L}]$ is equivalent to KdV; the spectrum of $\mathbb{L}$ is  time-invariant, and the study of the spectrum and eigenfunctions of the operator $\mathbb{L}$ and the dependence of the eigenfunctions on time $t$ allows one to solve the KdV equation via the  inverse scattering method \cite{FaddeevTakhtajan}.

For the goals of our paper, another form for the Lax pair for the KdV equation is more useful: 
\begin{equation}\label{eq:laxnew}
\mathbb{L} \left( {\mu}\right)
=
\begin{pmatrix}
  \partial_x & -1\\[0.3em] 
-\tfrac{\mu}{2} + \tfrac{1}{2}u(x,t) &   \partial_x
\end{pmatrix}
 \ ,  \quad 
\bar{\mathbb{P}}\!\left( {\mu}\right) =
\begin{pmatrix}
-\tfrac{1}{4}u_x &
 \mu + \tfrac{1}{2}u\\[0.6em]
-\tfrac{1}{4}u_{xx} - \tfrac{1}{4}u^2 - \tfrac{\mu}{4}u + \tfrac{\mu^2}{2} &
\tfrac{1}{4}u_x
\end{pmatrix}.
\end{equation}
We see that the operators act   now on  the  space of  2-dimensional-vector-valued   functions, i.e., functions of the form $(\psi_1(x), \psi_2(x))^T$. The entries of the second operator are differential polynomials. 
The Lax pair now depends on the  ``spectral'' parameter $\mu$, and gives the KdV equation for each  value of $\mu$.

The relation between the $\mathbb{L}$-operator of \eqref{eq:lax_old} and \eqref{eq:laxnew} is clear: for the  vectors  $(\psi_1(x), \psi_2(x))^T$  lying in the kernel of  $\mathbb{L}$ from  \eqref{eq:laxnew}, i.e., satisfying 
$$
\begin{pmatrix}
\partial_x & -1 \\ 
\tfrac{1}{2}\bigl(u -\mu \bigr) & \partial_x
\end{pmatrix}
\begin{pmatrix}\psi_1(x)\\ \psi_2(x)  \end{pmatrix}= 0,
$$
$\psi_2= \tfrac{\partial \psi_1}{\partial x}$ and  $\psi_1$  
satisfies $\tfrac{\partial^2  \psi_1}{\partial^2 x}  + \tfrac{1}{2} (u - \mu) \psi_1  =0.   $ We see that the 
function $\psi_1$ is an eingenfunction of   $\mathbb{L}$ from \eqref{eq:lax_old}
corresponding to eigenvalue $\mu$. Moreover, for every  eigenfunction $\psi$
of  $\mathbb{L}$ from \eqref{eq:lax_old} with eigenvalue $\mu$, the vector-function $\left(\psi, \tfrac{\partial}{\partial x} \psi\right)^T$ lies in the kernel of $\mathbb{L}$ given by \eqref{eq:laxnew}.

In our paper we construct a Lax pair for the so-called BKM systems and the corresponding hierarchy. This is a recently discovered~\cite{nijapp4} family of integrable systems of partial differential equations arising from Nijenhuis geometry, marking a new and challenging development at the intersection of differential geometry and mathematical physics.

These multicomponent, nonlinear, dispersive systems display features characteristic of physically relevant models: covariance, infinitely many conservation laws, hidden symmetries, and the existence of compactly supported and quasi-periodic solutions. Special cases recover several classical integrable models of central importance in physics and hydrodynamics, such as the KdV,  Camassa–Holm  and  Dullin-Gottwald-Holm equations, the Kaup–Boussinesq and Ito systems,  and less known coupled KdV, coupled Harry Dym, multicomponent Camassa–Holm equations, and the Marvan–Pavlov systems.

The initial construction of BKM systems was motivated by, and based on, the study of compatible multicomponent geometric Poisson structures in~\cite{nijapp2} (dispersionless case) and~\cite{nijapp3} (with dispersion). BKM systems admit (sufficiently many) conservation laws and symmetries of increasing order \cite{nijapp4}. Moreover, infinitely many essentially different “finite-gap  type” 
solutions of these equations can be constructed by reduction to finite-dimensional systems that are integrable in the sense of Liouville, whose solutions can be found by quadratures \cite{BKMreduction}, and, in certain cases, exactly.     
As mentioned above, the main result of the present paper is a construction of a Lax pair for the BKM systems. As in \eqref{eq:laxnew}, the Lax pairs are given by $2\times 2$ matrices whose entries are differential operators  acting on vector-valued functions; see Section~\ref{sec5}.

The BKM systems come with a natural continuous parameter, which we denote by $\lambda$, and are obtained as the first nontrivial term in a power expansion of a certain equation with respect to this parameter. Symmetries of the BKM system form a hierarchy corresponding to the next terms of this expansion. The Lax pair also depends on this parameter (in addition to the spectral parameter $\mu$), and the power expansions with respect to $\lambda$ provide Lax pairs for all members of the hierarchy.

The paper is organized as follows.
Section~\ref{sec1} provides the mathematical setup. We introduce/repeat the notion of formal differential series and formal evolutionary vector field, and  explain how studying  
evolutionary PDEs with  constraints can be reduced to studying formal evolutionary vector fields.
 BKM systems  and their symmetries are recalled in Section~\ref{sec_bkm}.  
    In Section \ref{sec4}, see  Theorem \ref{t4},  we rewrite the BKM system in an equivalent   parametric form. Theorem \ref{t5}  provides a    Lax pair for the system written in this new form. 
In Section \ref{sec5} we give Lax pairs for BKM of types I-IV and  describe an  algorithm for constructing Lax pairs for the hierarchies of symmetries. As examples, we construct Lax pairs  for KdV and Kaup-Bussinesq systems.
    Theorem \ref{t1} in Section \ref{sec3}  describes the normal forms  of the operators with respect to the conjugation by formal differential matrix series, and Theorem \ref{t2} describes the centralizer,    in the same category,  of a generic element. We use special case of these results to 
explain how one can use the constructed Lax pairs to find (hierarhies of) conservation laws for the BKM systems and explain the Lax pair geometry staying behind the constraint. Finally, in the conclusion, we sketch possible developments and applications of the results of the paper.


\section{Formal differential series and their relation to evolutionary PDEs with constraints}\label{sec1}
Let $u^1, \dots, u^n$ be coordinates on the $n$-dimensional ball $U^n$. The 
 space $J^\infty U^n$ is a bundle over $U^n$, each leaf of which is an infinite-dimensional vector space with coordinates $u^i_{x^j}$, $i = 1, \dots, n$, $j = 1, 2, \dots$. We refer to them as \emph{derivative coordinates} and use the convention $u^i_{x^0} = u^i$. Sometimes we write $u_{xx}, u_{xxx}$, etc. instead of $u_{x^2}, u_{x^3}$.

A function on $J^\infty U^n$ of the form
\begin{equation} \label{eq:h(u)}
h(u) (u^{i_1}_{x^{j_1}})^{n_1} \dots (u^{i_k}_{x^{j_k}})^{n_k},
\end{equation}
where $h(x)$ is a function (smooth, analytic, complex-valued, etc., depending on the situation; in this section we may assume that it is $C^\infty$-smooth) on $U^n$ and $n_i$ stands for the power, will be called {\it a differential monomial of degree} $j = j_1 n_1 + \dots + j_k n_k$.

A \emph{differential polynomial} is defined as a sum of differential monomials, only finitely many of which are different from zero. The \emph{degree} of a differential polynomial is the maximum of the degrees of its nonzero monomials. We say that the polynomial is \emph{homogeneous} if all its nonzero monomials have the same degree.

We define the operator $D$ on differential polynomials as
\begin{equation}\label{eq:D}
D = \sum_{j = 0}^\infty u^q_{x^{j + 1}} \pd{}{u^q_{x^j}}.
\end{equation}
This operator is well-defined, as for each differential polynomial $f$ there are only a finite number of nonzero terms in the formula for $Df$. The operator $D$ is called the \emph{complete derivative} (another standard notation for it in the literature is $\tfrac{d}{dx}$). The operator $D$ increases the differential degree of any differential polynomial by one, unless the differential polynomial has order $0$ and the corresponding coefficient is a constant. Clearly, $D\const=0.$

We denote by $\mathrm F (U^n)$ the ring of \emph{formal differential series}, that is, the series
$$
v = v_0 + v_1 + \dots,
$$
where each $v_j$ is a homogeneous differential polynomial of degree $j$. In particular, $v_0$ is a function on $U^n$. Addition and multiplication of the series are naturally defined. $D$ extends  to the formal differential series by $$D(v) = D(v_0) + D(v_1) + \dots \ .$$
We say that the series $v$ is \emph{formally invertible} if the function $v_0$ is not identically zero.
The reciprocal of an invertible series is given by the natural formula
\begin{eqnarray} \label{eq:rec}
(v_0+ v_1 + \dots)^{-1}&=& \tfrac{1}{v_0} \left(1-\tfrac{-v_1}{v_0}- \tfrac{-v_2}{v_0} - \dots\right)^{-1} \\ &=&   \tfrac{1}{v_0}
\left(1 + \left(\tfrac{-v_1}{v_0}+ \tfrac{-v_2}{v_0} + \dots\right)\ + \left(\tfrac{-v_1}{v_0}+ \tfrac{-v_2}{v_0} + \dots\right)^2  + \dots  \right).
\end{eqnarray}
Clearly, the expansion contains, for any $j$, only finitely many components of degree $j$, and is therefore a well-defined formal differential series. The kernel of the operation $D$ consists of constants, for which $v_0\in \mathbb{R}$ (or, later, $\mathbb{C}$) and all other $v_i\equiv 0$.
\begin{Comment}
\rm{
By their definition, differential polynomials are functions on $J^\infty \mathrm M^n$. To a certain extent and with a certain caution, one can work with differential series as with functions. Indeed, as explained above, one can multiply and add them, take reciprocals assuming $v_0   \ne  0$, differentiate them by $x$ (we denoted this operation by $D$ above) and Lie-differentiate with respect to a vector field.
Indeed, in all these operations, the result is still a differential series, each of whose coefficients is made from finitely many coefficients of the inputs. The operation $D$ satisfies the Leibniz rule: $D(vw)= wDv+ vDw$. Of course, unlike differential polynomials,
one cannot substitute a smooth vector-function $u(x)$ inside, as the result is not necessarily a convergent series.
}
\end{Comment}

Natural appearance of differential polynomials and differential series in our paper is related to a differential constraint (which is a system of polynomial relations on the jet bundle), which can be resolved in the class of formal differential series. Let us first give two examples.

\begin{Ex}\label{ex1}
\rm{
Consider the differential constraint in the form
$$
u = q + \frac{m}{2} q_{xx}.
$$
(As we will see in Example \ref{ex3}, this constraint comes from the Gottwald-Dullin-Holm system).
We substitute a differential series for $q$, \   $q = q_0 + q_1 + q_2 + \dots$, to get the recursion relation
$$ \begin{aligned} u & = q_0, \\ 0 & = q_1, \\ 0 & = q_{i + 2} + \frac{m}{2} (q_i)_{xx}, \quad i = 0, 1, 2, \dots. \end{aligned} $$
This recursive relation can be solved explicitly. The solution is a series $q$ whose coefficients are given by the following formulas:
\begin{equation}\label{eq:resolved} \textrm{$q_{2j + 1} = 0$ and $q_{2j} = (-1)^j \frac{m^j}{2^j} u_{x^{2j}}$ for $j = 0, 1, 2, \dots$ \ .}\end{equation}
}
\end{Ex}

The second example is of primary importance because of its direct relation to the BKM equations described below.

\begin{Ex} \label{ex2}
\rm{
Consider the differential constraint in the form
$$ 1 = m(\lambda) \Bigl( {w(\lambda)}_{xx} {w(\lambda)} - \tfrac{1}{2}\, {w(\lambda)}_x^2\Bigr) + \sigma(\lambda, u) {w(\lambda)}^2. $$
In this example $\lambda$ is a parameter;  $m(\lambda)$ is a given function (later, polynomial) in $\lambda$ only.  The function $\sigma$ depends on two variables, one of which is $\lambda$ (in later examples, the dependence on $\lambda$ will be polynomial or rational) and the other is $u(x)$.
Viewing $w(\lambda)$ as a differential series whose components depend on the parameter $\lambda$, that is, $w(\lambda) = w_0 (\lambda) + w_1(\lambda) + w_2(\lambda) + \dots$, we obtain the following relations
\begin{equation}\label{eq:relation} \begin{aligned} 1 & = \sigma(\lambda, u) w_0^2(\lambda), \\ 0 & = 2 \sigma(\lambda, u) w_0(\lambda) w_1(\lambda), \\ & \dots & \\ 0 & = m(\lambda)\Big( \sum_{j = 0}^k (w_j)_{xx} w_{k - j} - \frac{1}{2} \sum_{j = 0}^k (w_j)_x (w_{k - j})_x\Big) + \sigma(\lambda, u) \sum_{j = 1}^{k + 1} w_j w_{k +2 - j} + 2 \sigma(\lambda, u) w_0 w_{k + 2}, \end{aligned} \end{equation}
for $k = 0, 1, 2, \dots $\ .   Clearly, the first equation $1  = \sigma(\lambda, u) w_0^2(\lambda)$ can be solved with respect to $w_0$, and the last equation allows one to express $w_{k + 2}$ in terms of $\sigma(\lambda, u), m(\lambda)$ and $w_0(\lambda), \dots, w_{k + 1}(\lambda)$. Therefore, the system \eqref{eq:relation} provides a recursion formula for $w_k(\lambda)$, in which the dependence of $w_k$ on $\lambda$ is hidden in $m(\lambda)$ and $\sigma(u, \lambda)$.

Note also that, as in Example \ref{ex1}, the odd-numbered terms $w_{2i + 1}$ are identically zero. }
\end{Ex}

Following \cite{nijapp4}, by  \emph{a  formal evolutionary vector field }
$\partial_t$, we understand a derivation of the ring $\mathrm F (U^n)$, which commutes with $D$. The last property is clearly equivalent to
$$
\partial_t u^i_{x^j} = \partial_t D^j u^i = D^j\partial_t u^i,
$$
implying that any formal evolutionary vector field is uniquely defined by its action on the coordinates $u^i$. Therefore, one can write it as
\begin{equation} \label{eq:evpde}
\partial_t u^i = \xi^i [u, u_x, \dots], \quad i = 1, \dots, n,
\end{equation}
where $\xi^i \in \mathrm F (U^n)$. Note that if the series of each $\xi^i $ has a finite number of nonzero terms, \eqref{eq:evpde} is \emph{a standard evolutionary system of PDEs.} Clearly, any standard evolutionary system of PDEs is of Kovalevskaya type, implying that if its coefficients are real analytic, then for any initial local real-analytic curve $x\mapsto \hat u(x)$ there locally exists a unique real-analytic solution $u(x, t)$ with $u(x,0)= \hat u$. For formal evolutionary vector fields, a solution does not necessarily exist. Indeed, substituting $u(x,t),  \tfrac{\partial u(x,t)}{\partial x}, \tfrac{\partial^2  u(x,t)}{\partial x^2} , \dots \ $
instead of $u, u_x, u_{xx}, \dots $ in the right-hand side of the equation  $\tfrac{\partial u(x,t)}{\partial t}= \xi[u, u_x,\dots] (= \partial_t u)$ may diverge.

We call the elements of the quotient space $\mathcal F(U^n) = \mathrm F (U^n) / D(\mathrm F(U^n))$ \emph{formal functionals}. The projection map from $\mathrm F (U^n)$ to $\mathcal F (U^n)$ is denoted by $\int$. In particular, for any $F \in \mathcal F(U^n)$ there exists $f \in \mathrm F (U^n)$ such that $F = \int f$. The representation is not unique, that is, if $F = \int f = \int f'$ then $f - f' = D g$ for some $g \in \mathrm F (U^n)$. As $D$ commutes with any evolutionary equation, the derivation along this equation respects the quotient space; that is, we may view $\partial_t$ as a mapping $\partial_t : \mathcal F(U^n) \to \mathcal F(U^n)$.

We say that $f \in \mathcal F(U^n)$ is \emph{a conservation law} of $\partial_t$ if $\partial_t f = 0$ in the sense of $\mathcal F(U^n)$, i.e., if
$\partial_t f = D g
$
for a certain $g \in \mathrm F (U^n)$.

In our paper we deal with $n$-component evolutionary PDE-systems with    $s$ differential constraints. That is, we work with the PDE-systems of the form
\begin{equation}\label{eq:constraint}
\begin{aligned}
u_t & = \xi [u, u_x, \dots, u_{x^j}, q, q_x, \dots, q_{x^k}], \\
0 & = F_i (u, u_x, \dots, u_{x^l}, q, q_x, \dots, q_{x^r}), \quad i = 1, \dots, s.
\end{aligned}
\end{equation}
Here $u$ is an $n$-component vector-function and $q$ is an $s$-component vector-function. The solution of such a system $(u,q)$ consists of $  n +s $ functions of two variables $t, x$.

  We assume that the constraints can be resolved with respect to $q$ as a differential series (as we did in Examples \ref{ex1}, \ref{ex2}). This implies that each component $q^1, \dots, q^s$ of $q$ is a differential series $q^j = q^j_0 + q^j_1 + \dots$ of $u^1, \dots, u^n$.
Resolving the constraints  and substituting the results into the first  equation, we obtain a formal evolutionary vector field. If $(u,q)$ satisfies the initial equation, then $u$ satisfies the equation \eqref{eq:evpde} with this formal evolutionary vector field.

This approach has, in particular, the following advantages: for equations of the form \eqref{eq:evpde}, one can straightforwardly define the objects used in the theory of integrable systems, e.g., symmetries and conservation laws, mimicking the corresponding formulas in the case of standard evolutionary vector fields. For example, we say that two systems of the form \eqref{eq:constraint} are \emph{symmetries} of each other if their formal evolutionary vector fields commute, in a natural sense.
In naive terms, the last property means the following: take two systems of PDEs of the form \eqref{eq:constraint}, with the same $n$, but possibly different $s$, $\xi$ and $F_i$. We denote the variable ``$t$'' in the first system by $t_1$, and in the second by $t_2$, and similarly assume that the first system has $s_1$ constraints, and the second $s_2$. The variable ``$x$'' is the same in both systems.
The equations are mutual symmetries if the system of $2n+ s_1 +s_2$ equations made of both systems, but viewed as a system of $n+s_1+s_2$ unknown functions $u^1(x, t_1, t_2),\dots, u^1(x, t_1, t_2)$, $q_1^1(x, t_1, t_2),\dots,  q_2^{s_2}(x,t_1, t_2) $, is compatible in the sense that there exists no differential obstruction to the existence of solutions.

Note that the Lie bracket  of two evolutionary vector fields is again an evolutionary vector field, as the natural formula for it contains, for each degree, only finitely many terms.

\begin{Ex} \label{ex3}
\rm{
Consider the Dullin-Gottwald-Holm system \cite{gdh}
\begin{equation}\label{eq:DGH}
\begin{aligned}
u_t & = \frac{\gamma}{2} q_{xxx} + u q_x + \frac{1}{2} u_x q, \
u & = q + \frac{m}{2} q_{xx}.
\end{aligned}
\end{equation}
Here $\gamma$ and $ m$ are constants.  This is a system of two equations on two functions $u(t, x)$ and $q(t, x)$ (so $n=1=s$).
Note that for $\gamma = 0, m \neq 0$ one gets the Camassa-Holm equation, while $\gamma \neq 0, m = 0$ gives the   KdV equation (in more detail, the second equation implies $u=q$; substituting it into the first equation, we obtain the KdV equation).

The differential constraint in  \eqref{eq:DGH}  is the one considered in  Example \ref{ex1}, where we resolved it in the form of a formal differential series. Plugging this differential series into the first equation, one obtains a formal evolutionary equation of the form
\begin{equation}\label{eq:CH} u_t = \underbrace{\frac{3}{2} u u_x}_{\text{deg 1}} - \underbrace{\frac{m}{2}u u_{xx}}_{\text{deg 2}} + \underbrace{\frac{\gamma}{2}u_{xxx} - \frac{m}{2} u_x u_{xx}}_{\text{deg 3}} + \dots \ . \end{equation}
Above, as usual, dots stand for the terms of higher order.

A side unrelated observation is that, in view of \eqref{eq:CH} and from the viewpoint of formal evolutionary equations, the Dullin-Gottwald-Holm equation is a formal perturbation, in the formal parameter $m$, of
the  KdV  equation (the coefficients are different from those of \eqref{eq:KdV} but one can change the coefficients by multiplying $u$ with a constant and changing linearly
the variables $x$ and $t$).
}
\end{Ex}



\section{BKM systems and their symmetries}\label{sec_bkm}
Choose numbers  $n\in \mathbb{N}$,   $r, n_0, \dots, n_r \in \mathbb{N}\cup \{0\}$ and $l_1, \dots, l_r\in \mathbb{N}$, such that
$$
n = n_0 + n_1 + \dots + n_r \quad \text{and} \quad d:= n_0 - l_1 n_1 - \dots - l_r n_r \geq 0.
$$
Choose a polynomial $m(\mu) = m_0 + m_1 \mu  + \dots + m_d \mu^d$ of degree $\le d$.

These $2r + 1$ numbers $n_0, \dots, l_r$ and the $d + 1$ coefficients $m_0, \dots, m_d$ of the polynomial are parameters of the construction. There will be one more parameter of the construction, a number $\lambda \in \mathbb{C} \cup {\infty}$; the type I–IV of the BKM system, which will be introduced below, will depend on whether $\lambda$ is infinity or not, and whether $\lambda$ is a zero of $m$ or not.

Next, take the product of $r + 1$ discs $\mathrm{M}^n = U_0^{n_0} \times \dots \times U_r^{n_r}$ of the indicated dimensions
and consider the operator field $L$ in the form $L = L_0 \oplus \dots \oplus L_r$, where each $L_i$ is a differentially non-degenerate Nijenhuis operator on the corresponding disc.
The definition of Nijenhuis operators and differential nondegeneracy is given in \cite{nij1}. For our paper, it is sufficient to know that for a certain coordinate system
(we call it \emph{companion coordinates} and denote by c.c.) a differentially nondegenerate operator is given by the second matrix of \eqref{eq:opL}.
In the coordinate system
$$ \underbrace{u^1_0, \dots, u^{n_0}_0}_{\text{c.c. on $U_0$}}, \underbrace{u^1_{1}, \dots, u^{n_1}_{1}}_{\text{c.c. on $U_1$}}, \dots, \underbrace{u^1_{r}, \dots, u^{n_r}_{r}}_{\text{c.c. on $U_r$}} $$
``made’’ from companion coordinates on the blocks, the operator field $L$ is block-diagonal with blocks in the companion form:
\begin{equation} \label{eq:opL} L = \left(\begin{array}{cccc} L_0 & 0 & \dots & 0 \\ 0 & L_1 & \dots & 0 \\ & & \ddots & \\ 0 & 0 & \dots & L_r \\ \end{array}\right), \quad \text{where} \quad L_i = \left( \begin{array}{ccccc} u_i^1 & 1 & 0 & \dots & 0 \\ u_i^2 & 0 & 1 & \dots & 0 \\ & & & \ddots & \\ u_i^{n_i - 1} & 0 & 0 & \dots & 1 \\ u_i^{n_i} & 0 & 0 & \dots & 0 \end{array}\right), \quad i = 0, 1, \dots, r. \end{equation}

Note, however, that differential nondegeneracy does not depend on a choice of coordinates. In a different coordinate system, the operator \eqref{eq:opL} may take a very different form, so the corresponding BKM system may look very different from the one constructed using \eqref{eq:opL}. 

Next, take the following function $\sigma$ depending on both the coordinates $u$ and the parameter $\mu$.
\begin{equation}\label{eq:sigma}
\sigma(u, \mu) = \frac{\det(L_0 - { \mu } \Id)}{\left(\det(L_1 - { \mu} \Id) \right)^{l_1} \cdots \left(\det(L_r - { \mu }\Id)\right)^{l_r}}
\end{equation}

\begin{Comment}
\rm{

As we see below,
there are two closely related but visually different classes of BKM systems. The roles of $m$ and of the function $\sigma(u,\mu)$ are played by
\begin{equation}\label{eq:barsigma}
\bar \sigma(u, \mu) = (-\mu)^d\sigma\left(u, \frac{1}{\mu}\right) \quad \text{and} \quad
\bar m(\mu) = (- \mu)^d\ m\left( \frac{1}{\mu}\right).
\end{equation}

The function $\bar m$ is a polynomial in $\mu$ of degree $\le d$ and the function $\sigma$ is rational in $\mu$ with coefficients depending on $u$.

Geometrically, the parameter $\mu$ and the later introduced parameter $\lambda$ can be viewed as projective objects, and indeed the polynomial $\bar m$ is just the polynomial $m$ viewed as a homogeneous polynomial on $\mathbb{C}P^1 = \bar{\mathbb{C}}= \mathbb{C}\cup {\infty}$ written in the chart containing $\infty$ in the coordinate $1/\mu$. In particular, we say that $\mu = \infty$ is \emph{a root of} $m(\lambda)$ if $\bar m(0) = 0$. Each class of BKM systems is divided into two subclasses, depending on whether $\lambda$ is a root of the polynomial $m$ and $\bar m$.}
\end{Comment}
Now consider the vector field $\zeta_0$ on $\mathrm M^n$, defined by the relation
\begin{equation} \label{eq:zeta}
\mathcal L_{\zeta_0} \det(L_0 - \mu \Id) = 1, \quad \mathcal L_{\zeta_0} \det (L_i - \mu \Id) = 0, \quad i = 1, \dots, r .
\end{equation}
Above, $\mathcal L$ denotes the Lie derivative. As the coefficients of the characteristic polynomial of $L$ are functionally independent,
\eqref{eq:zeta} uniquely determines $\zeta_0$.
Set $\zeta = p(L) \zeta_0$, where $$p(t) = (\det(L_1 - t \Id))^{l_1} \dots (\det(L_r - t \Id))^{l_r} m(t).$$

\begin{Ex}
Assume $r=0$ and take $L=L_0$ in the companion form \eqref{eq:opL}. Then,
\begin{equation}\label{eq:zetaexplicit} \zeta_0 = (-1)^{n+1} \begin{pmatrix} 0 \\ \vdots \\ 0 \\ 1 \end{pmatrix}\ , \quad \zeta = (-1)^{n+1} \begin{pmatrix} m_{n-1} + m_n u^1 \\ m_{n-2} + m_n u^2 \\ \vdots \\ m_0 + m_n u^n \end{pmatrix}. \end{equation}
\end{Ex}
The BKM system is the following family of $n$-component evolutionary PDE systems with one constraint:
\begin{equation}\label{BKM} \begin{aligned} u_{t_\lambda} & = {q(\lambda)}_{xxx} \Bigl(L - \lambda \operatorname{Id}\Bigr)^{-1}\zeta + {q(\lambda)} \Bigl(L - \lambda \operatorname{Id}\Bigr)^{-1}u_x, \\ 1 & = m(\lambda) \Bigl( {q(\lambda)}_{xx} {q(\lambda)} - \tfrac{1}{2}\, {q(\lambda)}_x^2\Bigr) + \sigma(\lambda, u) {q(\lambda)}^2. \end{aligned} \end{equation}
Above, $\lambda$ is a parameter which parametrises the systems within the family.
Note that there is a natural gauge equivalence on the class of systems \eqref{BKM}: changing $\lambda$ to $\lambda+\lambda_0$ and the polynomial $m(\lambda)$ to $m(\lambda-\lambda_0)$ gives the same BKM system. Moreover,
as mentioned above, the parameter $\lambda$ can be viewed as a projective parameter and we allow it to take the value $\infty$. The coordinate in a neighborhood of infinity (in $\bar{\mathbb{C}}= \mathbb{C}\cup {\infty} $) is $\mu =  \tfrac{1}{\lambda}$, and in this
coordinate the BKM system has the following form:
\begin{equation}\label{BKMinf} \begin{aligned} u_{t_\mu} & = {\bar q(\mu)}_{xxx} \Bigl(\operatorname{Id} - \mu L\Bigr)^{-1}\zeta + {\bar q(\mu)} \Bigl(\operatorname{Id} - \mu L\Bigr)^{-1}u_x, \\ 1 & = \bar m(\mu) \Bigl( {\bar q(\mu)}_{xx} {\bar q(\mu)} - \tfrac{1}{2}\, {\bar q(\mu)}_x^2\Bigr) + \bar \sigma(\mu, u) {\bar q(\mu)}^2. \end{aligned} \end{equation}

  Let us explain that the system \eqref{BKM} coincides with \eqref{BKMinf} for $\mu= \tfrac{1}{\lambda}$ (after an appropriate reparameterization of the time). The point $\mu=0$ is special and corresponds to $\lambda=\infty $ in \eqref{BKM}.

Substituting $\lambda=  1/\mu  $ in the first equation of
\eqref{BKM} gives

$$ \mu^{d/2-1} u_{t}  =  -\mu^{d/2}\left({q(\lambda)}_{xxx} \left(\operatorname{Id}  - \mu L \right)^{-1}\zeta + {q(\lambda)} \left(\operatorname{Id}  - \mu L \right)^{-1}u_x\right).$$
After the rescaling  of time and of the  function $q$ by the formula    \begin{equation}\label{eq:transition}  t_\mu=  \mu^{d/2-1}t  \  \textrm{   and    }  \  
\bar q(\mu)=-\mu^{d/2} q(1/\mu), \end{equation} The equation \eqref{BKM} gives \eqref{BKMinf}.

Let us comment on the notation $t_\lambda, t_\mu$ in  (\ref{BKM}, \ref{BKMinf}) above. In formulas \eqref{BKM} and
\eqref{BKMinf} we assume that $\lambda$ and $\mu$ are fixed. The variable $t_\lambda$ (resp. $t_\mu)$ is just
the ``time'' $t$, which we considered above.
Similarly, the notation $q(\lambda)$ and $\bar q(\mu)$ does not carry any special meaning in \eqref{BKM} and \eqref{BKMinf},
as $\lambda$ and $\mu$ are fixed, so $q(\lambda)$ and $\bar q(\mu)$ are just unknown functions of $(x,t_\lambda)$ and of $(x, t_\mu)$.
The reason for introducing the subscript $\lambda$ (resp., $\mu$) and the dependence $q(\lambda)$ on $\lambda$ and $\bar q(\mu)$ on $\mu$ will become clear below,
see Facts \ref{fact:1} and \ref{fact:2}: Fact \ref{fact:1} shows that for functions $q_1= q(\lambda_1)$ and $q_2= q(\lambda_2)$ the corresponding evolutionary equations with constraint are mutual symmetries, and Fact \ref{fact:2} shows how to produce commutative symmetries, of increasing order, of the BKM equations.

{   By \emph{BKM I} system we understand the system   \eqref{BKM} such that $\lambda$ is not a root of the polynomial $m$.  In the case when $\lambda$ is a root of $m$,   the second equation  of \eqref{BKM} reads
$q = \sqrt{\sigma(\lambda, u)}$. Plugging this into the first equation of \eqref{BKM}, we obtain the following evolutionary system (no constraint!) of $n$ PDEs:
\begin{equation}\label{BKM:II}
 u_{t_\lambda}  = \left(\tfrac{1}{\sqrt{ \sigma(\lambda, u)}} \right)_{xxx}\Bigl(L - \lambda \operatorname{Id}\Bigr)^{-1}\zeta + \tfrac{1}{\sqrt{ \sigma(\lambda, u)}} \Bigl(L - \lambda \operatorname{Id}\Bigr)^{-1}u_x.
\end{equation}
We will call the system \eqref{BKM:II}  \emph{BKM II system}.

As was shown in \cite{nijapp4}, the BKM systems are integrable in the sense that they possess an infinite family of symmetries (and also of conservation laws) of increasing order. The main observation leading to this is that the symmetries of the system \eqref{BKM} are just the system itself, but with other values of the parameters:

\begin{Fact}[\cite{nijapp4}]  \label{fact:1}
For any two values $\lambda_1, \lambda_2$ of $\lambda$, the corresponding systems \eqref{BKM} are symmetries of each other.
\end{Fact}

Assume now $\mu = 0$ (  $   \Longleftrightarrow \ \lambda=\infty$). In this case the second equation of \eqref{eq:barsigma} gives $\bar \sigma = 1$. Then, the constraint in \eqref{BKMinf} reads

\begin{equation}\label{eq:mu0}1 = \bar m(0) \Bigl( {\bar q(0)}_{xx} {\bar q(0)} - \tfrac{1}{2}\, {\bar q(0)}_x^2\Bigr) + {\bar q(0)}^2.\end{equation}
 This ODE  implies $\bar q= \textrm{const}= 1$: indeed, differentiating it by $x$  we obtain 
  $$  \bar m(0)  \bar q_{xxx} \bar q  +  2 \bar q_x  \bar q = 0  \  \Longrightarrow  \  \bar m(0) \bar q_{xx} +   \bar q{=} \const   \   {\Longrightarrow}  \   -\bar q_x\bar q  + 2 \bar q_x \bar q=0         \   \stackrel{\eqref{BKMinf}}{\Longrightarrow}  \ 
 \bar q= \const =1.$$

As $\bar q=1$, the first equation of \eqref{BKMinf} is just $u_{t_\infty}=  u_x$ and is not very interesting. It is still, of course, a symmetry of \eqref{BKM} for every $\lambda$.

We see that the systems with different values of the parameter $\lambda$, including $\lambda=\infty$, provide a large family of mutual symmetries. In most publications on integrable systems, one equation or one system of equations is selected and one looks for mathematical methods that allow one to analyse it, in particular those coming from the theory of integrable systems, e.g., hierarchies of commuting symmetries of increasing order and of conservation laws of increasing order. This selection may be explained as follows: one starts with a system which describes a physically interesting phenomenon and then looks for methods, in particular for those coming from the theory of integrable systems, to study it.
This is one reason why in many interesting and physically relevant examples commuting symmetries and conservation laws come in a hierarchy. An additional advantage of passing to hierarchies is as follows: for most $\lambda$, the BKM systems are systems with a differential constraint. As a system of commutative symmetries we can choose a hierarchy of evolutionary systems with no constraint. As will be clear below, such hierarchies correspond to roots of the polynomial $m$. Note that even the case $m(\lambda)= m_0= \const\ne 0$ can be treated in this way, as in this case $m(\infty)= \bar m(0)=0$, and the corresponding hierarchy at $\lambda=\infty$ still provides a system of evolutionary equations with no constraint.

Moreover, the first terms of hierarchies of the BKM system at $\lambda= \infty$, i.e., the system \eqref{BKMinf} at $\mu=0$, (BKM III and BKM IV, see below) contain a large family of interesting evolutionary integrable systems with constraints and with no constraint, containing previously studied integrable systems including the KdV, the Camassa-Holm and the Kaup–Boussinesq equations.

As proved in \cite{nijapp4}, we obtain a hierarchy of commuting symmetries of increasing order by the following procedure:

\begin{enumerate}
    \item Choose $\lambda_0\in \mathbb{C} \cup \{\infty\}$. 
    \item If  $\lambda_0\ne \infty$,  consider the formal decomposition, in the formal parameter $\varepsilon$,  of the left-hand side
    \begin{equation}\label{eq:partialt}
    \partial_{t_{\lambda_0 + \varepsilon}} u = \partial_{t_{\lambda_0, 0}}u + \varepsilon \partial_{t_{\lambda_0, 1}}u + \varepsilon^2 \partial_{t_{\lambda_0, 2}}u + \dots.
 \end{equation}
{\      For $\lambda = \infty$ ($\Longleftrightarrow  \   \mu=0$), we consider the following formal decomposition: 
 \begin{equation}
 \label{eq:partialtinfty}\partial_{t_{\infty+ \varepsilon}} u =  \partial_{t_{\infty, 0}}u + \varepsilon \partial_{t_{\infty, 1}}u + \varepsilon^2 \partial_{t_{\infty, 2}}u + \dots \ .\end{equation}}
    \item Consider the following formal decomposition of the  the constraint:
$$
    \begin{aligned}
     1 & = m(\lambda_0 + \varepsilon) \Bigl( {q(\lambda_0+\varepsilon)}_{xx} {q(\lambda_0 + \varepsilon)} - \tfrac{1}{2}\, {q(\lambda_0 + \varepsilon)}_x^2\Bigr) + \sigma(\lambda_0 + \varepsilon, u) {q(\lambda_0 + \varepsilon)}^2 = \\
     & =  Q_0 + \varepsilon Q_1 + \varepsilon^2 Q_2 + \dots  
    \end{aligned}
    $$
    and the corresponding decomposition at $\mu=0$
    $$
    \begin{aligned}
     1 & = \bar m(\varepsilon) \Bigl( {\bar q(\varepsilon)}_{xx} {\bar q(\varepsilon)} - \tfrac{1}{2}\, {\bar q(\varepsilon)}_x^2\Bigr) + \bar \sigma(\varepsilon, u) {\bar q( \varepsilon)}^2 = \\
     & = \bar Q_0 + \varepsilon \bar Q_1 + \varepsilon^2 \bar Q_2 + \dots.
    \end{aligned}
    $$
  Coefficients of the series in $\varepsilon$ of $q(\varepsilon)$ and of $\bar q(\varepsilon)$ are formal series. Consequently, $Q_i$ and $\bar Q_i$ are formal series. However, if $\lambda_0$ is a root of the polynomial $m(\lambda)$ or $\bar m(0)=0$, then the resulting coefficients are differential polynomials.

\item   Take the BKM equation \eqref{BKM} (\eqref{BKMinf}, resp.), substitute $\lambda_0+\varepsilon$ for $\lambda$, replace $t_{\lambda_0+\varepsilon}$ by \eqref{eq:partialt} (resp., replace $t_{\infty+\varepsilon}$ by \eqref{eq:partialtinfty}), and equate the coefficients of the same power in $\varepsilon$. Equating the $i$th power in $\varepsilon$ gives us a system of the following form
    $$ 
    \partial_{\lambda, i}  = \xi[u, q_0, \dots, q_i]    \  \ \ (\textrm{resp., \   \ }  \partial_{\infty, i}  = \xi[u, \bar q_0, \dots, \bar q_i]  )
    $$
    Next, add to the system $i+1$ constraints $Q_0=1,\ Q_1=0,\ \dots,\ Q_i=0$ (resp., $\bar Q_0=0,\ \bar Q_1=0,\ \dots,\ \bar Q_i=0$). For each $i$, we obtain a formal evolutionary equation with $i+1$ constraints. However, again, if $\lambda_0$ is a root of $m$ (resp., $\bar m(0)=0$), then it is a standard evolutionary equation with no constraints.
    
\end{enumerate}

}
\begin{Fact}[\cite{nijapp4}] \label{fact:2}
For every $i_1, i_2$, the systems above are mutual symmetries. Moreover, all of them are symmetries of the BKM system \eqref{BKM} and its barred version \eqref{BKMinf}.
\end{Fact}

\begin{Ex} {\rm \label{ex:3.2}
Let us calculate the first two terms of the hierarchy at $\lambda=\infty$ ($\Longleftrightarrow\ \mu=0$). We clearly have
$$\bar m(\varepsilon) =(-1)^d ( m_d + m_{d-1}\varepsilon + \dots )\ , \ \ \bar \sigma(\varepsilon)= 1 - \varepsilon
(\tr(L_0)- \sum_{i=1}^rl_i\tr(L_i))+ \dots \ , \ \bar q= \bar q_0 + \varepsilon \bar q_1 + \dots \ , 
$$ 
where the dots denote higher-order terms in $\varepsilon$. Substituting this into the second equation of \eqref{BKMinf}, we see that the coefficients of the zeroth-order terms in $\varepsilon$ yield the equation
\begin{equation}\label{BKMIIItmpa}
1 = m_d({{({\bar q}_0})}_{xx}{\bar q_0} 
- \tfrac{1}{2} (({\bar q}_0)_x)^2)+ ({\bar q}_0)^2,   
\end{equation}
We already discussed this equation in \eqref{eq:mu0} and saw that it implies $\bar q_0 = 1$.
In particular, we have $\bar q = 1 + \bar q_1 \varepsilon + \dots$ . Moreover, in view of $\bar q_0 = 1$, the first equation of \eqref{BKMinf} gives us $u_x = u_{t_\infty}$.
Next, let us equate the terms of first degree in $\varepsilon$: we obtain
$$0  = 2 \bar q_1  +  (-1)^d m_d (\bar q_1)_{xx} - \tr L_0 + \sum_{i=1}^r l_i \tr L_i.$$
This is the constraint in the first-degree term of the hierarchy. Substituting $q=1+\varepsilon \bar q_1$ and $(\Id-\varepsilon L)^{-1}=(\Id+\varepsilon L+\dots)$ into the first equation of \eqref{BKMinf}, we obtain
\begin{equation}\label{BKMIIItmp}
u_{t_{1,\infty}}= (\bar q_1)_{xxx}(\Id + \varepsilon L + \dots) \zeta +  (L  + \bar q_1 \Id) u_x. 
\end{equation}
}
    \end{Ex}
Next, depending on whether $\mu=0$ is a root of $\bar m$, formulas \eqref{BKMIIItmpa} and \eqref{BKMIIItmp} can be naturally divided into the following two cases, referred to as BKM III and BKM IV.

\textbf{BKM III}. Assume that $\infty$ is not a root of the polynomial $m$, that is, assume that $m_d\ne 0$. Then \eqref{BKMIIItmpa} and \eqref{BKMIIItmp} give us the following system with a constraint:
\begin{equation}
\label{eq:BKM3}
\begin{aligned}
u_{t_\infty} &= q_{xxx} \zeta + (L + q\,\Id) u_x,\\
0 & = 2q  +  (-1)^d m_d q_{xx} - \tr L_0 + l_1 \tr L_1 + \dots + l_r \tr L_r.
\end{aligned}
\end{equation}
 
\textbf{BKM IV}. Now take $\lambda=\infty$ ($\Longleftrightarrow\ \mu=0$) and assume $\bar m(0)=0$. In this case, the second equation of \eqref{eq:BKM3} gives
$q=\frac{1}{2}\bigl(\tr L_0-l_1\tr L_1-\dots-l_r\tr L_r\bigr)$,
and we obtain
$$
u_{t_\infty} = \tfrac{1}{2} (\tr L_0 - l_1 \tr L_1 - \dots - l_r \tr L_r)_{xxx} \, \zeta + \left(L + \tfrac{1}{2} (\tr L_0 - l_1 \tr L_1 - \dots - l_r \tr L_r) \cdot \Id\right)\, u_x.
$$
This is an $n$-component evolutionary equation (with no constraints). For $r=0$, we obtain the following simpler equation
$$u_{t_\infty} = \tfrac{1}{2} (\tr L)_{xxx}  \zeta + \left(L + \tfrac{1}{2} \tr L \cdot \Id\right) u_x.$$

\begin{Ex}\label{ex:3.3}
Consider BKM IV systems corresponding to the case $n=1$, $r=0$, and $m=\textrm{const}=1$. In this case, $\sigma(\lambda)=u-\lambda$, so $\zeta_0=\zeta=\tfrac{\partial}{\partial u}$ (i.e., its only component equals $1$), and $\bar\sigma(\mu)=-\mu\left(u-\tfrac{1}{\mu}\right)=1-\mu u$, and we obtain the  KdV equation \eqref{eq:KdV}:
$$
u_{t_\infty} = \tfrac{1}{2} u_{xxx}  +\tfrac{3}{2} u u_x.
$$
\end{Ex}

\begin{Comment}
\rm{
The construction above does not distinguish between evolutionary and formal evolutionary equations. This allows us to find (usual) evolutionary symmetries for formal evolutionary systems, and for evolutionary systems with constraints. For example, the Dullin--Gottwald--Holm and Camassa--Holm equations are BKM III equations with $m(\lambda)$ of the first degree. Therefore, they are evolutionary equations with a constraint. If we choose a root of $m$ as $\lambda_0$, we obtain evolutionary symmetries with no constraints, which are symmetries for both the Dullin--Gottwald--Holm and Camassa--Holm systems. We also see that the Dullin--Gottwald--Holm and Camassa--Holm systems are symmetries of each other (of course, this fact was known, in different terms, before).

}    
\end{Comment}



\section{Parametric form of BKM systems and the corresponding  Lax pairs}\label{sec4}
A Lax pair representation is essentially a rewriting of a system of ODEs or (evolutionary or formal evolutionary)  PDEs  into the form \eqref{eq:Lax}. We proceed as follows: Theorem \ref{t4}  rewrites  the initial BKM systems (\ref{BKM}, \ref{BKMinf}) in another equivalent form, for which we will construct the Lax pair in Theorem \ref{t5}.

\begin{Theorem}\label{t4}
Consider a BKM system, and let $\sigma(\mu), \bar\sigma(\mu)$ be the functions given in \eqref{eq:sigma} and \eqref{eq:barsigma}. Then the following statements hold:
\begin{enumerate}
    \item
    For every $\lambda\ne\infty$, the BKM system \eqref{BKM} is equivalent to the parametric system
     \begin{equation}\label{eq:parametricform1}
    \begin{aligned}
    \partial_{t_\lambda} \sigma(\mu) & = \frac{1}{\mu - \lambda} \Bigg( m(\mu) q_{xxx} (\lambda) + 2 \sigma(\mu) q_x(\lambda) + \sigma_x(\mu) q(\lambda)\Bigg), \\
    1 & = m(\lambda) \Bigg( q_{xx}(\lambda) q(\lambda) - \frac{1}{2} q^2_x(\lambda)\Bigg) + \sigma(\lambda) q^2(\lambda) 
    \end{aligned}
    \end{equation}
    for all  $\mu$.
    \item For every $\lambda\ne\infty$, the BKM system \eqref{BKMinf} is equivalent to the parametric system
 \begin{equation}\label{eq:parametricform2}
    \begin{aligned}
    \partial_{t_\lambda} \bar \sigma(\mu) & = \frac{ \mu}{\mu - \lambda} \Bigg( \bar m(\mu) \bar q_{xxx} (\lambda) + 2 \bar \sigma(\mu) \bar q_x(\lambda) + \bar \sigma_x(\mu) \bar q(\lambda)\Bigg), \\
    1 & = \bar m(\lambda) \Bigg( \bar q_{xx}(\lambda) \bar q(\lambda) - \frac{1}{2} \bar q^2_x(\lambda)\Bigg) + \bar \sigma(\lambda) \bar q^2(\lambda) 
    \end{aligned}
   \end{equation}
    for all $\mu$.
\end{enumerate}
\end{Theorem}
Let us give some additional explanations. First, note that the first equation of the BKM system with $\lambda\ne\infty$ is well-defined also for $\lambda=\mu$. Indeed, the dependence on $\mu$ is analytic, and at $\mu=\lambda$ the term
 $ m(\mu) q_{xxx} (\lambda) + 2 \sigma(\mu) q_x(\lambda) + \sigma_x(\mu) q(\lambda)$ is divisible by $\mu-\lambda$ because of the second equation: to see this, differentiate the second equation with respect to $x$ at $\mu=\lambda$. A similar explanation works in the case  \eqref{eq:parametricform2}.

Let us also comment on the word ``equivalent''. The BKM system \eqref{BKM} is a system of $n+1$ equations for the $n+1$ unknown functions $u^1,\dots,u^n,\ q$. The system in Theorem \ref{t4} contains infinitely many equations, since one can choose infinitely many values of the parameter $\mu$. The dependence on the parameter $\mu$ is rational, and by choosing $n$ arbitrary generic values of $\mu$ we obtain a system of $n+1$ equations for the $n+1$ unknown functions $u^1,\dots,u^n,\ q$. Theorem \ref{t4} states that this system is equivalent to the corresponding BKM system. In particular, it is not important which (generic) values of $\mu$ we have chosen.

\begin{Ex} \label{ex:KB} {\rm
As we will see below, the parametric form of the BKM equations appears to be more convenient for certain tasks, so it can make sense to also write other elements of the BKM hierarchy in parametric form. We will do it at $\lambda_0=\infty$, assuming $n=2$, $r=0$, and $m(\lambda)\equiv 1$ (so $\bar m(\lambda)=\lambda^2,\ \bar\sigma(\lambda)=1-\lambda u^1-\lambda^2 u^2$). To this end, we substitute $q(\lambda)=1+\lambda q_1+\lambda^2 q_2+\dots$ into the constraint $1=\bar m(\lambda) q_x(\lambda)^2+\bar\sigma(\lambda) q(\lambda)^2$. We obtain

\[
 q(\lambda)=1+\frac{\lambda}{2}u^1+\lambda^2\Bigl(\frac{1}{2}u^2+\frac38(u^1)^2\Bigr)+\dots  \   \  \Longleftrightarrow  \  \  q_1=\frac{1}{2}u^1\ ,  \  \   q_2=\frac{1}{2}u^2+\frac{3}{8}(u^1)^2. 
\]
Substituting this $q$ in the equation (we recall that $t_{\infty,0}=x$ and $t_{\infty,1}= t$)  $$(\partial_{x}  + \lambda \partial_{t} + \dots) \bar \sigma(\mu)  = \frac{ \mu}{\mu - \lambda} \Bigg( \bar m(\mu) \bar q_{xxx} (\lambda) + 2 \bar \sigma(\mu) \bar q_x(\lambda) + \bar \sigma_x(\mu) \bar q(\lambda)\Bigg) $$ and equating the terms that are linear in $\lambda$ gives
\[
-\mu u_{1t}-\mu^2u_{2t}
=\frac{\mu^2}{2}u_{1xxx}
-\frac{3}{2}\mu u_1u_{1x}
-\mu^2u_2u_{1x}
-\mu\Bigl(1+\frac{\mu u_1}{2}\Bigr)u_{2x},
\]
which is  clearly equivalent to \begin{equation}\label{eq:KB}
\begin{aligned}
u^1_{t} &= u^2_{x}+\frac{3}{2}\,u^1u^1_{x},\\
u^2_{t} &= -\frac{1}{2}\,u^1_{xxx}+\frac{1}{2}\,u^1u^2_{x}+u^2u^1_{x}.
\end{aligned}
\end{equation}
This system is one of the equivalent forms of the Kaup--Boussinesq (= dispersive water wave) system.
    } 
\end{Ex}

Finally, to avoid a possible misunderstanding, we would like to comment on the change of notation in \eqref{eq:parametricform2} with respect to \eqref{BKMinf}. Namely, the parameter $\lambda$ in \eqref{eq:parametricform2} plays the role of the parameter $\mu$ in \eqref{BKMinf}. In particular, $\lambda=0$ in \eqref{eq:parametricform2} corresponds to $\lambda=\infty$ in \eqref{eq:parametricform1}. The parameter $\mu$ in \eqref{eq:parametricform2} has no analogue in \eqref{BKMinf}.

\begin{proof}[Proof of Theorem \ref{t4}]
We start with the following Lemma:

\begin{Lemma}\label{m1}
For the Nijenhuis operator $L$ given by \eqref{eq:opL} and the functions $\sigma$ and $\bar\sigma$ given by \eqref{eq:sigma} and \eqref{eq:barsigma}, the following holds for any $\mu$:
$$
(L^* - \lambda \Id)^{-1} \ddd \sigma(\mu) = \frac{1}{\mu - \lambda} \Bigg(\ddd \sigma(\mu) -  \frac{\sigma(\mu)}{\sigma(\lambda)} \ddd \sigma(\lambda)\Bigg)
$$
and
$$
(\Id - \lambda L^*)^{-1} \ddd \bar \sigma(\mu) = \frac{\mu}{\mu - \lambda} \Bigg(\ddd \bar \sigma(\mu) -  \frac{\bar \sigma(\mu)}{\bar \sigma(\lambda)} \ddd \bar \sigma(\lambda)\Bigg)
$$
\end{Lemma}
(Again, as in Theorem \ref{t4}, both formulas in the Lemma make sense and remain correct also for $\mu=\lambda$.)

\begin{proof}
By \cite{nij1}, the differentials of the determinant and the trace of any Nijenhuis operator $L$ satisfy the following relation: $L^*\ddd\det L=\det L,\ddd\tr L$. Since polynomials of Nijenhuis operators are Nijenhuis operators, for an arbitrary power $s$ we obtain
$$
L^* \ddd (\det L)^s = s (\det L)^s \ddd \tr L.
$$
Hence, for an arbitrary $\lambda$ we obtain, for the function $\sigma(u,\lambda)$ defined in \eqref{eq:sigma},
\begin{equation}\label{sm1}
(L^* - \lambda \Id) \ddd \sigma(u, \lambda) = \sigma(u, \lambda)\ddd (\tr L_0 - l_1 \tr L_1 - \dots - l_r \tr L_r).    
\end{equation}
We denote the second factor on the right-hand side, which does not depend on $\lambda$, by $f$, and obtain $(L^*-\lambda\Id)\ddd\sigma(\lambda)=\sigma(\lambda)\ddd f$. This formula holds for any value of $\lambda$; taking $\lambda\ne\mu$, we get
$$
(L^* - \lambda \Id) \frac{\ddd \sigma(\lambda)}{\sigma(\lambda)} = (L^* - \mu \Id) \frac{\ddd \sigma(\mu)}{\sigma(\mu)} = \ddd f.
$$
We rewrite it as
$$
(L^* - \lambda \Id) \frac{\ddd \sigma(\lambda)}{\sigma(\lambda)} = (L^* - \lambda \Id) \frac{\ddd \sigma(\mu)}{\sigma(\mu)} + (\lambda - \mu) \frac{\ddd \sigma(\mu)}{\sigma(\mu)}.
$$
Multiplying both sides by $(L^*-\lambda\Id)^{-1}$ and rearranging the terms, we obtain the first statement of the lemma. Substituting $\lambda\to\lambda^{-1}$, $\mu\to\mu^{-1}$, and multiplying both sides by $(-\mu)^d$, we see that the left-hand side takes the form
$$
(-\mu)^d (L^* - \lambda^{-1} \Id)^{-1} \ddd \sigma(\mu^{-1}) = - \lambda (\Id - \lambda L^*)^{-1} \ddd \bar \sigma(\mu).
$$
The right hand side  takes the  form
$$
\begin{aligned}
(-\mu)^d\frac{1}{\mu^{-1} - \lambda^{-1}} \Bigg(\ddd \sigma(\mu^{-1}) -  \frac{\sigma(\mu^{-1})}{\sigma(\lambda^{-1})}\ddd \sigma(\lambda^{-1})\Bigg) & = \frac{\lambda \mu}{\lambda - \mu} \Bigg(\ddd \bar \sigma(\mu) -  \frac{\bar \sigma(\mu)}{(-1)^d \sigma(\lambda^{-1})} (-1)^d\ddd \sigma(\lambda^{-1})\Bigg) = \\
& = \frac{\lambda \mu}{\lambda - \mu} \Bigg(\ddd \bar \sigma(\mu) -  \frac{\bar \sigma(\mu)}{\bar \sigma(\lambda)} \ddd \bar \sigma(\lambda)\Bigg).
\end{aligned}
$$
Dividing both sides by $-\lambda$, we obtain the second statement of Lemma.
\end{proof}

\begin{Lemma}\label{m2}
Under the assumptions of Lemma \ref{m1}, and for $\zeta$ given by \eqref{eq:zeta}, we have
$$
\mathcal L_{\zeta} \sigma(\mu) = m(\lambda) - (-1)^d m_d \sigma(\lambda).
$$
and
$$
\mathcal L_{\zeta} \bar \sigma(\mu) = \bar m(\lambda) - (-1)^d m_d \bar \sigma(\lambda).
$$
\end{Lemma}
\begin{proof}
We work in the coordinate system in which $L$ is given by \eqref{eq:opL}; it is built from companion coordinates for the operators $L_i$. We obtain
\begin{equation}\label{sm2}
\begin{aligned}
L^* \ddd u^i_0 & = u^i_0 \ddd u^1_0 + \ddd u^{i + 1}_0, \quad i = 1, \dots, n_0 - 1, \\
L^* \ddd u^{n_0}_0 & = u^{n_0}_0 \ddd u^1_0.
\end{aligned}    
\end{equation}
Next, using   \eqref{sm1} we obtain 
$$
L^* \ddd \sigma(\mu) = \mu \ddd \sigma(\mu) + \sigma(\mu) \ddd f,
$$
which implies $\mathcal L_{L\zeta_0} \sigma(\mu) = \mu$. Combining this with  \eqref{sm2}, we obtain 
$$
\mathcal L_{L^k \zeta_0} \sigma(\mu) = \mu^k, \quad k = 0, \dots, d - 1 \quad \text{and} \quad \mathcal L_{L^d \zeta_0} = u^1_0 \mu^{n_0 - 1} + \dots + u^{n_0}_0. 
$$
Finally, we get  
$$
\mathcal L_\zeta \sigma(\mu) = m_0 + m_1 \mu + \dots + m_{d - 1} \mu^{d - 1} + m_d(u^1_0 \mu^{n_0 - 1} + \dots + u^{n_0}_0) = m(\mu) - (-1)^d m_d \sigma(\mu).
$$
The first statement of the Lemma is proved. The second statement is obtained after substituting $\mu\to 1/\mu$ and multiplying both sides by $(-\mu)^d$.
\end{proof}

We are now able to prove Theorem \ref{t4}. Differentiating the constraint, we obtain
$$
\bigg(m(\lambda) q_{xxx}(\lambda) + 2 \sigma(\lambda) q_x(\lambda) + \sigma_x(\lambda) q(\lambda)\bigg) q(\lambda) = 0.
$$
Then, the expression in brackets in the formula above  vanishes. Combining this with Lemmas \ref{m1} and \ref{m2}, we obtain
$$
\begin{aligned}
& \partial_{t_\lambda} \sigma(\mu) = \ddd \sigma(\mu) (u_{t_\lambda}) = \frac{1}{\mu - \lambda} q(\lambda)_{xxx} \Bigg(m(\mu) - (-1)^d \sigma(\mu)  -  \frac{\sigma(\mu)}{\sigma(\lambda)} \big(m(\mu) - (-1)^d \sigma(\mu)\big)\Bigg) + \\
& + \frac{1}{\mu - \lambda} q(\lambda) \Bigg(\sigma_x(\mu) - \frac{\sigma(\mu)}{\sigma(\lambda)} \sigma_x(\lambda)\Bigg) = \frac{1}{\mu - \lambda} \Big( m(\mu) q_{xxx}(\lambda) + 2 \sigma(\mu) q_x(\lambda) + \sigma_x(\mu) q(\lambda)\Big).
\end{aligned}
$$
Then \eqref{BKM} implies the first equation in \eqref{eq:parametricform1}. The converse statement is also evident since the dynamics of the phase variables can be reconstructed from the dynamics of $\sigma(\mu)$ for all values of the parameter~$\mu$. 

To prove the second statement, we recall that the systems \eqref{BKM} and \eqref{BKMinf} are equivalent, and the substitution establishing this equivalence is \eqref{eq:transition}. Using this substitution, we obtain \eqref{eq:parametricform2} from \eqref{eq:parametricform1}.
\end{proof}

We are now ready to formulate our main result, namely, the Lax pair for BKM systems. Consider $C^\infty_2(\mathbb R)=C^\infty(\mathbb{R};\mathbb{C}^2)$ and the following operators $L,\mathbb{P}:C^\infty_2 \to C^\infty_2$:

$$
L(\mu) = \left(\begin{array}{cc}
     0 & 1  \\
     - \frac{1}{2}\frac{\sigma(\mu)}{m(\mu)} & 0 
\end{array}\right), \quad \mathbb{P} (\mu, \lambda) = \left( \begin{array}{cc}
     - \frac{1}{2} q_x(\lambda) & q(\lambda)  \\
     - \frac{1}{2}q_{xx} (\lambda) -\frac{1}{2} \frac{\sigma(\mu)}{m(\mu)} q(\lambda) & \frac{1}{2}q_x(\lambda)
\end{array}\right).
$$
and
$$
\bar L(\mu) = \left(\begin{array}{cc}
     0 & 1  \\
     - \frac{1}{2}\frac{\bar \sigma(\mu)}{\bar m(\mu)} & 0 
\end{array}\right), \quad \bar{\mathbb{P}} (\mu, \lambda) = \left( \begin{array}{cc}
     - \frac{1}{2} \bar q_x(\lambda) & \bar q(\lambda)  \\
     - \frac{1}{2} \bar q_{xx} (\lambda) - \frac{1}{2} \frac{\bar \sigma(\mu)}{\bar m(\mu)} \bar q(\lambda) & \frac{1}{2}\bar q_x(\lambda)
\end{array}\right).
$$
The components of $L$ depend on a single parameter $\mu$, while the components of $\mathbb{P}$ depend on two parameters, $\mu$ and $\lambda$. In both cases, the dependence on the parameter $\mu$ is rational.
We denote by $D:C^\infty_2\to C^\infty_2$ the operator $D=\begin{pmatrix}D & 0\\ 0 & D\end{pmatrix}$, i.e., acting on $\Psi=(\psi_1,\psi_2)^T$ by
$$D\begin{pmatrix} \psi_1\\ \psi_2\end{pmatrix}=\begin{pmatrix} D\psi_1\\ D\psi_2\end{pmatrix}.$$

\begin{Theorem}\label{t5}
Let $\mathbb{L}(\mu)=D-L(\mu)$ and $\bar{\mathbb{L}}(\mu)=D-\bar L(\mu)$, where $L(\mu)$, $\bar L(\mu)$, $\mathbb{P}(\mu,\lambda)$, and $\bar{\mathbb{P}}(\mu,\lambda)$ are as above. Then, for every $\lambda\ne\infty$, the BKM system \eqref{BKM} is equivalent to the parametric Lax equation

\begin{equation}
    \label{eq:laxBKM}
\partial_{t_\lambda} \mathbb{L}(\mu) = \frac{1}{\mu - \lambda} [\mathbb{P}(\mu, \lambda),  \mathbb{L}(\mu)]
\end{equation}
with condition that for $m(\lambda) \neq 0$ the right hand  side has no pole in $\mu = \lambda$ and normalization 
\begin{equation}\label{eq:normalisation1}
2m(\lambda) \det \mathbb{P}(\lambda, \lambda) = 1. 
\end{equation}
Moreover,  the for every $\lambda\ne \infty$, the  BKM system \eqref{BKMinf} is equivalent to the parametric Lax equation 
\begin{equation}
    \label{eq:laxBKMinf}
\partial_{t_\lambda} \bar {\mathbb{ L}} (\mu) = \frac{\mu}{\mu - \lambda} [\bar{\mathbb{P}}(\mu, \lambda),\bar {\mathbb{L}} (\mu) ]
\end{equation}
with the similar  condition (no pole for $\lambda = \mu$ if $\bar m(\lambda) \neq 0$) and normalization 
\begin{equation}\label{eq:normalisation2}
2 \bar m(\lambda) \det \bar{\mathbb{P}}(\lambda, \lambda) = 1.
\end{equation}
\end{Theorem}
Let us give additional explanations. First of all, as the formulas are analytic in $\mu$, the condition that the right-hand side of \eqref{eq:laxBKM} has no pole means that it vanishes for $\mu=\lambda$. A similar reformulation also holds for \eqref{eq:laxBKMinf}. Second, as we will see in the proof, the conditions \eqref{eq:normalisation1} and \eqref{eq:normalisation2} are just the constraints in \eqref{BKM} and \eqref{BKMinf}. See also the discussion at the end of Section \ref{sec4}. Finally, note that (as also discussed after Theorem \ref{t4}), the role of $\lambda$ and $\mu$ in the equations \eqref{eq:laxBKMinf} and \eqref{eq:normalisation2} is different from that in \eqref{BKMinf}.

Note also that the factors $\tfrac{1}{\lambda-\mu}$ and $\tfrac{\mu}{\lambda-\mu}$ in the parametric Lax equations \eqref{eq:laxBKM} and \eqref{eq:laxBKMII} can be ``hidden'' in $\mathbb{P}$ and $\bar{\mathbb{P}}$, which turns the parametric equation into the usual Lax equation \eqref{eq:laxnew}.

\begin{proof} The proof is straightforward: we calculate the parametric Lax equation \eqref{eq:laxBKM} and compare it with the equivalent form of the BKM equations from Theorem \ref{t4}. Note that, in these calculations, the form of the functions $\sigma(\mu)$ and $m(\mu)$, and their barred analogues, is not important (though of course we use that $\sigma$ depends on $\mu,x,t$ and $m$ depends on $\mu$ only). Let us give details of the calculations. Clearly,
$$
\partial_{t_\lambda} \mathbb{L}=  \partial_{t_\lambda} \left(\begin{array}{cc}
     0 & 1  \\
     - \frac{1}{2}\frac{\sigma(\mu)}{m(\mu)} & 0 
\end{array}\right) = \left(\begin{array}{cc}
     0 & 0  \\
     - \frac{1}{2}\frac{\partial_{t_\lambda}\sigma(\mu)}{m(\mu)} & 0 
\end{array}\right).
$$
Next, clearly, $[\mathbb L,\mathbb{P}]=[D,\mathbb{P}]-\mathbb{L}\mathbb{P}+\mathbb{P}\mathbb{L}$. We carefully calculate all three terms on the right: 
$$
[D, \mathbb{P}]= D \mathbb{P}- \mathbb{P} D=   \mathbb{P}_x = \left( \begin{array}{cc}
     - \frac{1}{2} q_{xx}(\lambda) & q_x(\lambda)  \\
     - \frac{1}{2}q_{xxx} (\lambda) - \frac{1}{2} \frac{\sigma_x(\mu)}{m(\mu)} q(\lambda) - \frac{1}{2} \frac{\sigma(\mu)}{m(\mu)} q_x(\lambda) & \frac{1}{2}q_{xx}(\lambda)
\end{array}\right).
$$
$$
 \begin{aligned}
L\mathbb{P}  &= \left(\begin{array}{cc}
     0 & 1  \\
     - \frac{1}{2}\frac{\sigma(\mu)}{m(\mu)} & 0 
\end{array}\right) \left( \begin{array}{cc}
     - \frac{1}{2} q_x(\lambda) & q(\lambda)  \\
     - \frac{1}{2}q_{xx} (\lambda) -\frac{1}{2} \frac{\sigma(\mu)}{m(\mu)} q(\lambda) & \frac{1}{2}q_x(\lambda)
\end{array}\right) = \\
& = \left( \begin{array}{cc}
     - \frac{1}{2}w_{xx} (\lambda) -\frac{1}{2} \frac{\sigma(\mu)}{m(\mu)} q(\lambda) &  \frac{1}{2} q_x(\lambda) \\
     \frac{1}{4}\frac{\sigma(\mu)}{m(\mu)} q_x(\lambda) &  - \frac{1}{2}\frac{\sigma(\mu)}{m(\mu)} q(\lambda)
\end{array}\right),   
\end{aligned}
$$
$$
\begin{aligned}
 \mathbb{P}L&= \left( \begin{array}{cc}
     - \frac{1}{2} q_x(\lambda) & q(\lambda)  \\
     - \frac{1}{2}q_{xx} (\lambda) -\frac{1}{2} \frac{\sigma(\mu)}{m(\mu)} q(\lambda) & \frac{1}{2}q_x(\lambda)
\end{array}\right) \left(\begin{array}{cc}
     0 & 1  \\
     - \frac{1}{2}\frac{\sigma(\mu)}{m(\mu)} & 0 
\end{array}\right) = \\
& = \left(\begin{array}{cc}
     - \frac{1}{2}\frac{\sigma(\mu)}{m(\mu)} q(\lambda) & - \frac{1}{2} q_x(\lambda) \\
     - \frac{1}{4}\frac{\sigma(\mu)}{m(\mu)} q_x(\lambda) & - \frac{1}{2}q_{xx} (\lambda) -\frac{1}{2} \frac{\sigma(\mu)}{m(\mu)} q(\lambda)
\end{array}\right).     
\end{aligned}
$$
  Combining these three terms, we obtain 
$$
\begin{aligned}
\relax
[\mathbb{L}(\mu), \mathbb{P}(\mu, \lambda)] & = \mathbb{P}_x(\mu, \lambda) - L(\mu) \mathbb{P}(\mu, \lambda) + \mathbb{P}(\mu, \lambda) L(\mu) = \\
& = \left( \begin{array}{cc}
     0 & 0  \\
     - \frac{1}{2}q_{xxx} (\lambda) - \frac{1}{2} \frac{\sigma_x(\mu)}{m(\mu)} q(\lambda) - \frac{1}{2} \frac{2 \sigma(\mu)}{m(\mu)} q_x(\lambda) & 0
\end{array}\right).    
\end{aligned}
$$
Finally, formula $\partial_{t_\lambda} \mathbb{L}(\mu) - \frac{1}{\mu - \lambda}[\mathbb{L}(\mu), \mathbb{P}(\mu, \lambda)]$ is equivalent to 
$$
- \frac{1}{2} \frac{\partial_{t_\lambda} \sigma(\mu)}{m(\mu)} = \frac{1}{\mu - \lambda} \Bigg(- \frac{1}{2}q_{xxx} (\lambda) - \frac{1}{2} \frac{\sigma_x(\mu)}{m(\mu)} q(\lambda) - \frac{1}{2} \frac{2 \sigma(\mu)}{m(\mu)} q_x(\lambda)\Bigg).
$$
Multiplying both sides by $-2m(\mu)$, we arrive at the first equation \eqref{eq:parametricform1} of Theorem \ref{t4}. The no-pole condition at $\lambda=\mu$ is equivalent to $[\mathbb L(\lambda),\mathbb{P}(\lambda,\lambda)]=0$. Theorem \ref{t2} implies that $\det \mathbb{P}(\lambda,\lambda)=c(\lambda)$, i.e., a constant depending on the parameter. The normalization condition reads 
$$
m(\lambda) (q_{xx}\big(\lambda) q(\lambda) - \frac{1}{2} q_x(\lambda)\big) + \sigma(\lambda) q^2(\lambda) = 1.
$$
This is exactly the differential constraint in \eqref{eq:parametricform1}. The first statement of Theorem \ref{t4} is proved. In order to prove the second statement, we note that the calculations in the proof did not use the special form of $\sigma$ and $m$, and would work for any functions $\sigma(\mu,u)$ and $m(\mu)$. As $\bar L$ and $\bar{\mathbb{P}}$ are given by essentially the same formulas as $L$ and ${\mathbb{P}}$, with the only difference that all objects inside are barred, the Lax equation for $\bar{\mathbb{L}}$ and $\bar{\mathbb{P}}$ gives the barred version of the Lax equation for ${\mathbb{L}}$ and ${\mathbb{P}}$. The difference between the parametric Lax equations \eqref{eq:laxBKM} and \eqref{eq:laxBKMinf} reflects the different  factors $\tfrac{1}{\mu-\lambda}$ and $\tfrac{\mu}{\mu-\lambda}$ in the equations \eqref{eq:parametricform1} and \eqref{eq:parametricform2}.
\end{proof}


\section{Lax pairs for four types of BKM systems and examples}\label{sec5}
\subsection{Lax pairs for BKM I and BKM II systems}
Theorem \ref{t5} allows one to obtain the Lax pairs for four type of BKM equations and hierarchies of their  symmetries. The Lax pairs for BKM I and II are already  given by Theorem \ref{t5}. 

\textbf{ Lax pair for BKM I. } Fix  $\lambda\ne \infty $ and assume $m(\lambda) \neq 0$. We get
\begin{equation}\label{eq:laxBKMI}
L(\mu) = \left(\begin{array}{cc}
     0 & 1  \\
     - \frac{1}{2} \frac{\sigma(\mu)}{m(\mu)} & 0 
\end{array}\right), \quad  \mathbb{P}(\mu) = \left( \begin{array}{cc}
     - \frac{1}{2} q_x & q  \\
     - \frac{1}{2} q_{xx} - \frac{1}{2} \frac{\sigma(\mu)}{m(\mu)} q & \frac{1}{2}q_x
\end{array}\right)
\end{equation}
The corresponding Lax pair is $\mathbb L(\mu)=D-L(\mu)$ and $\mathbb P(\mu)$. The parametric Lax equation \eqref{eq:laxBKM} is equivalent to the BKM I equations \eqref{BKM}. Here $\mu$ is a ``spectral'' parameter, and the ``normalization'' $m(\lambda)\det \mathbb P(\lambda)=1$ is the constraint.

\textbf{Lax pair for BKM II. } Fix  $\lambda\ne \infty$ and assume now  $m(\lambda) = 0$. Then, the constraint $m(\lambda) \det \mathbb P(\lambda) = 1$ can be resolved with respect to $q$ and  gives  
$q = \sigma(\lambda, u)^{-1/2}$. Consequently,    \eqref{eq:laxBKMI} reads   
\begin{equation}\label{eq:laxBKMII}
L(\mu) = \left(\begin{array}{cc}
     0 & 1  \\
     - \frac{1}{2}\frac{\sigma(\mu)}{m(\mu)} & 0 
\end{array}\right), \quad\mathbb{P} (\mu) = \left( \begin{array}{cc}
     - \frac{1}{2}\Big(\frac{1}{\sqrt{\sigma(\lambda)}}\Big)_x & \frac{1}{\sqrt{\sigma(\lambda)}}  \\
     - \frac{1}{2}\Big(\frac{1}{\sqrt{\sigma(\lambda)}}\Big)_{xx} - \frac{1}{2} \frac{\sigma(\mu)}{\sqrt{\sigma(\lambda)}} & \frac{1}{2}\Big(\frac{1}{\sqrt{\sigma(\lambda)}}\Big)_x
\end{array}\right).
\end{equation}
The corresponding Lax pair is $\mathbb L(\mu) = D - L(\mu), \mathbb{P} (\mu)$. The parametric Lax equation \eqref{eq:laxBKM} is equivalent to the   BKM II equation \eqref{BKM:II}.

\subsection{Lax pairs for the hierarchies and for  BKM III and BKM IV  systems. } The construction of the hierarchies of 
symmetries was described   in Section \ref{sec_bkm}. For the BKM equation \eqref{BKM},  it is  based on the formal expantion  \eqref{eq:partialt} of $\partial_{t_{\lambda_0+ \varepsilon}}$ and substitution power series $q=q_0+ \varepsilon q_1+ \dots$  in  the BKM equation and in the corresponding constraints, and equating coefficients of the obtained formal power series in  $\varepsilon$. One obtains   Lax pairs for the hierarchy by a similar procedure: 
we replace  $\partial_{t_{\lambda_0+ \varepsilon}}$ by \eqref{eq:partialt},  
$\lambda$ by $\lambda_0+ \varepsilon$, $q$ by  $q_0+ \varepsilon q_1+ \dots$.     
As the operator $\mathbb{L}$ in the Lax pair does not depend on $\lambda$ and the right hand side of the Lax equation 
linearly depends on $\tfrac{1}{\mu-\lambda}\mathbb{P}$,   we obtain the Lax pairs with the  same $\mathbb{L}$  as before. In these Lax pairs, 
the new $\mathbb{P}$-operators  are   the coefficients  of the initial $\mathbb{P}$, after  these substitutions, at the corresponding powers of $\varepsilon$.  We note that we can use the following expansion of  the factor $\tfrac{1}{\mu- \lambda_0 - \varepsilon}$: 
$$\frac{1}{\mu- \lambda_0 - \varepsilon}=   \frac{1}{\mu- \lambda_0}\left(1 + \frac{\varepsilon}{\mu-\lambda_0} + \left(\frac{\varepsilon}{\mu-\lambda_0}\right)^2 +\dots\right).     $$
We need to multiply by this formal series in $\varepsilon$ the matrix    
$$  \mathbb{P} (\mu, \lambda_0+\varepsilon) = \left( \begin{array}{cc}
     - \frac{1}{2} (q_0+ \varepsilon q_1 + \dots)_x+ & q_0+ \varepsilon q_1 +\dots  \\
     - \frac{1}{2}(q_0+ \varepsilon q_1  + \dots  )_{xx} -\frac{1}{2} \frac{\sigma(\mu)}{m(\mu)} (q_0+ \varepsilon q_1 + \dots) & \frac{1}{2}(q_0 +\varepsilon q_1+ \dots)_x
\end{array}\right). $$ The terms at $\varepsilon^k$ will give us  the $\mathbb{P}$-matrix for the $k$th element of the hierarchy.

In the case $\lambda_0=\infty$ (that is, for the BKM equation \eqref{BKMinf} at $\mu=0$), the procedure is similar. Of course, in this case we substitute $\lambda=\varepsilon$ in the formula for $\bar{\mathbb{P}}$, and there is an additional factor $\mu$ in the parametric Lax equation.

\begin{Ex} \rm{ Let us calculate the Lax pair for  first two equations of the BKM hierarchy  corresponding to \eqref{BKMinf}. We have
$$\begin{aligned}
\bar{\mathbb{P}} (\mu, \varepsilon) & = \left(1 + \tfrac{\varepsilon}{\mu} + \dots\right) \begin{pmatrix} 
     - \tfrac{1}{2}(1+ \varepsilon  \bar q_1+ \dots)_x & 1+ \varepsilon \bar q_1  + \dots  \\
     - \tfrac{1}{2} (1+ \varepsilon \bar q_1+ \dots)_{xx} - \tfrac{1}{2} \tfrac{\bar \sigma(\mu)}{\bar m(\mu)} (1 + \varepsilon \bar q_1+ \dots ) & \tfrac{1}{2}(1+ \bar q_1+ \dots)_x 
\end{pmatrix}  \\
& =   \begin{pmatrix}
0 & 1 \\[0.4em]
-\tfrac{1}{2}\tfrac{\bar \sigma(\mu)}{\bar m(\mu)} & 0
\end{pmatrix}  + \varepsilon\begin{pmatrix}
-\tfrac{1}{2}\,(\bar q_1)_x 
&
\bar q_1 + \tfrac{1}{\mu}
\\ 
-\tfrac{1}{2}\,{(\bar{q}_1)}_{xx}
-\tfrac{1}{2}\tfrac{\bar \sigma(\mu)}{\bar m(\mu)}\,\bar q_1
-\tfrac{1}{2\mu}\tfrac{\sigma(\mu)}{m(\mu)}
&
\tfrac{1}{2}\,(\bar q_{1})_x
\end{pmatrix} + \dots \ .
\end{aligned}
$$
The zero-order term gives us the Lax pair for the first equation of the BKM hierarchy at $\lambda_0=\infty$, which, as we know, is simply $\sigma(\mu)_t=\sigma(\mu)_x$. The first-order term gives us the $\mathbb{P}$-operator for the first nontrivial element of the BKM hierarchy at $\lambda_0=\infty$, which, as we know, corresponds to the BKM III and BKM IV equations. }
\end{Ex}

As Lax pairs for the BKM equations are the main topic and title of the paper, we write the formulas once more, with the cosmetic change $\bar q_1=\bar q$ below.

\textbf{Lax pair for BKM III. }
\begin{equation}\label{eq:laxBKMIII} 
 \bar L(\mu)  = \left(\begin{array}{cc}
     0 & 1  \\
     - \tfrac{1}{2}\tfrac{\bar \sigma(\mu)}{\bar m(\mu)} & 0 
\end{array}\right),  \quad \bar{\mathbb{P}}(\mu) = \left( \begin{array}{cc}
     - \tfrac{1}{2} \bar q_x & \frac{1}{\mu} + \bar q  \\
     - \tfrac{1}{2\mu} \frac{\bar \sigma(\mu)}{\bar m(\mu)} - \tfrac{1}{2} \bar q_{xx} - \tfrac{1}{2} \tfrac{\bar \sigma(\mu)}{\bar m(\mu)} \bar q & \frac{1}{2}\bar q_x.
\end{array}\right)
\end{equation}
The constraint is  $2\bar q + (-1)^d m_d \bar q_{xx} = \tr L_0 - l_1 \tr L_1 - \dots - l_r \tr L_r$. The corresponding Lax pair is $\bar{\mathbb L}(\mu) = D - \bar L(\mu), \bar{\mathbb{P}}(\mu)$.

\textbf{Lax pair for BKM IV. }
\begin{equation}\label{eq:laxBKMIV} 
\bar L(\mu) = \left(\begin{array}{cc}
     0 & 1  \\
     - \tfrac{1}{2}\frac{\bar \sigma(\mu)}{\bar m(\mu)} & 0 
\end{array}\right),  \quad \bar{\mathbb{P}} (\mu) = \left( \begin{array}{cc}
     - \frac{1}{2} \bar q_x & \frac{1}{\mu} + \bar q  \\
     - \frac{1}{2\mu} \frac{\bar \sigma(\mu)}{\bar m(\mu)} - \frac{1}{2} \bar q_{xx} - \frac{1}{2} \frac{\bar \sigma(\mu)}{\bar m(\mu)} \bar q & \frac{1}{2}\bar q_x 
\end{array}\right)
\end{equation}
with  $\bar q = \tfrac{1}{2}\left(\tr L_0 - l_1 \tr L_1 - \dots - l_r \tr L_r\right)$. The corresponding Lax pair is $ \bar{\mathbb L}(\mu) = D - \bar L(\mu), \bar{\mathbb{P}}(\mu)$.

We would like to emphasize that the Lax equation for the BKM III and BKM IV systems with the above Lax pair is the usual Lax equation \eqref{eq:Lax}, not the ``parametric'' one (having an additional factor on the right-hand side, as in, e.g., \eqref{eq:laxBKM} or \eqref{eq:laxBKMinf}).

\begin{Ex} \label{ex:laxKdV} \rm{As demonstrated in Example \ref{ex:3.3}, the KdV equation \eqref{eq:KdV} is a BKM IV equation corresponding to the parameters $r=0$, $n=1$, and $\bar m(\mu)=-\mu$ ($\Longleftrightarrow\ m=1$). We have already seen in Example \ref{ex:3.3} that, in this case, $\bar q=\tfrac{1}{2}u$ and $\bar \sigma=1-\mu u$. Substituting this into \eqref{eq:laxBKMIV} gives $$ 
\bar{\mathbb{L}}(\mu)
=
\begin{pmatrix}
 D & -1\\ 
-\tfrac{1}{2\mu} + \tfrac{1}{2}u(x,t) &  D
\end{pmatrix}
 \ ,  \quad \bar{\mathbb{P}} (\mu) =\begin{pmatrix}
-\tfrac{1}{4}u_x &
 \tfrac{1}{\mu} + \tfrac{1}{2}u\\[0.6em]
-\tfrac{1}{4}u_{xx} - \tfrac{1}{4}u^2 - \tfrac{1}{4\mu}u + \tfrac{1}{2\mu^2} &
\tfrac{1}{4}u_x
\end{pmatrix}.$$
  
We see that (up to the natural change of $\mu$) it coincides with \eqref{eq:laxnew}.

}    
\end{Ex}

\begin{Ex}
\rm{ Let us calculate the Lax pair for the Kaup--Boussinesq system, which we already considered in Example \ref{ex:KB}. It corresponds to the parameters $n=2$, $r=0$, and $m(\mu)=1$. As explained in Example \ref{ex:KB}, $$ \bar \sigma(\mu)=1-\mu u^1-\mu^2 u^2,\quad \bar m(\mu)=\mu^2,\quad \bar q=\tfrac{1}{2}u^1.$$ Substituting these into \eqref{eq:laxBKMIV} gives

\[
\bar{\mathbb{P}} (\mu)=
\begin{pmatrix}
-\frac14\,u^1_{x} & \frac12\,u^1+\frac{1}{\mu}\\
-\frac{1}{4}\,u^1_{xx}-\frac{1}{2}\,\frac{1 - \mu u^1  - \mu^2 u^2}{\mu^2 }
\left(\frac{1}{2}\,u^1+\frac{1}{\mu}\right)
& \frac{1}{4}\,u^1_{x}
\end{pmatrix}.
\]
The Lax equation $\tfrac{\partial }{\partial t}  \bar{\mathbb{L}}=  [  \bar{\mathbb{P}}  , \bar{\mathbb{L}}]$ gives us  the KB-system \eqref{eq:KB}. 

}    
\end{Ex}

\section{Construction of conservation  laws for BKM systems from the Lax pair representation}\label{sec3}

Fix $m$ and denote by  $C^{\infty}_m=  C^{\infty} (\mathbb R; {\mathbb{C}^m})$ 
the linear space of smooth $m-$vector valued functions of $x\in \mathbb{R}$, that is the functions of the form $\psi = (\psi^1, \dots, \psi^m)^T$ with $\psi_i(x)\in \mathbb{C} $. We will need the following two operators  acting on this space. The first one is $D: C^{\infty}_m  \to C^{\infty}_m  $ is just  
the component-wise differentiation of $\psi$, with respect to the variable   $x$:  $(D\psi)^i = \psi^i_x$.  Two-dimensional version of this operator was used in  Theorem \ref{t5} and Section \ref{sec5}.

The second class of operators are matrix operators $L: C^{\infty}_m  \to C^{\infty}_m  $ whose  components $L^i_j$ are  (complex-valued)   functions of $x$. Such an operator $L$  acts on $\psi$  by the standard   matrix multiplication,  $(L\psi)^i = \sum_{s=1}^m L^i_s \psi^s$.  We denote the linear space of such operators as $\operatorname{Mat}(m)$.

\begin{Ex}
\rm{
If $m = 1$ then we get $D = \pd{}{x}$ and $\operatorname{Mat}(1)$ consists of multiplications by a   function $\ell(x)$.  
}    
\end{Ex}

Further we assume that $m > 1$.  In this case,  the operators $L$ form a non-commutative algebra of matrix-valued functions. 

\begin{Theorem}\label{t1}
Assume $m>1$. Consider the operator $\mathbb{L}:C^{\infty}_m \to C^{\infty}_m $ of the form $\mathbb{L}=D+L$, where $D$ is a derivation in $x$ and $L\in\textrm(m)$, as above. Assume that $L$ has a simple spectrum for all $x$ (i.e., it has $m$ distinct eigenvalues). Then there exists a matrix $C$ such that:
\begin{enumerate}
\item The components $C^i_j$ are differential series in the components $L^i_j$ of $L$.
\item The following identity holds:
\begin{equation}
\label{eq:identity}
C^{-1}(D+L)C = D+L_{\mathrm{diag}},
\end{equation}
where $L_{\mathrm{diag}}$ is diagonal, and its diagonal elements $l_i$ are differential polynomials in the components $L^i_j$.
\end{enumerate}
\end{Theorem}
\begin{proof}
We start with a simple  linear-algebraic Lemma.

\begin{Lemma}\label{l1}
Assume that $L$ has a simple spectrum. Then any matrix $M$ can be written in the form
\begin{equation}\label{eq:decomp}
M = [A, L] + B,
\end{equation}
where $[B, M] = 0$.
\end{Lemma}
\begin{proof}
 Without loss of generality, $L$ can be assumed to be diagonal. For diagonal $L$, the decomposition \eqref{eq:decomp} is just the decomposition of $M$ into a diagonal matrix $A$ and a matrix $B'$ with zeros on the diagonal. It is known, and easy to prove, that any such $B'$ can be written as $B'=[A,L]$ for a certain matrix $A$.
\end{proof}

To prove the theorem, we need to establish the existence of matrices $C=C_0+C_1+\dots$ and $L_{\mathrm{diag}}=S_0+S_1+\dots$, where $C_i$ and $S_i$ are matrices whose entries are homogeneous differential polynomials of degree $i$, and, in addition, each $S_i$ is diagonal.

We choose any $C_0\in\operatorname{Mat}(m)$ such that $C_0^{-1}LC_0$ is diagonal; the existence of such a matrix $C_0$ follows from the simplicity of the spectrum. As $C_0$ is nondegenerate, we can equivalently rewrite \eqref{eq:identity} as
$$
(D+L)(C_0+C_1+C_2+\dots)=(C_0+C_1+C_2+\dots)(D+S_0+S_1+\dots).
$$
Indeed, the formula \eqref{eq:rec} can be straightforwardly generalized to matrices. By our choice of $C_0$, the above equation is automatically satisfied at the level of the $(i=0)$-order terms, since $LC_0=C_0S_0$. For each order $i>0$, we obtain the equation

$$
L C_i  + (C_{i - 1})_x = C_i S_0  + C_{i - 1} S_1 + \dots + C_0 S_i.
$$
Multiplying both sides by $C_0^{-1}$ and rearranging the terms, we obtain
\begin{equation} \label{eq:commS}
[S_0, C^{-1}_0 C_i] - S_i = C_0^{-1} \big( C_{i - 1} S_1 + \dots + C_1 S_{i - 1}\big).
\end{equation}
Recall that the entries of $C_i$ are homogeneous differential polynomials of degree $i$, so they are sums of terms of the form \eqref{eq:h(u)}. For any $i_1,\dots,i_k$ and $n_1,\dots,n_k$ such that $i_1n_1+\dots+i_kn_k=i$, the corresponding monomial
$u_{x^{i_1}}^{n_1}\cdots u_{x^{i_k}}^{n_k}$
appears linearly on the left-hand side of \eqref{eq:commS}, with a coefficient in $\operatorname{Mat}(m)$. Applying Lemma \ref{l1} to the coefficient of this monomial in
$M:=C_0^{-1}\big(C_{i-1}S_1+\dots+C_1S_{i-1}\big)$,
we can find $A$ and $B$ (for $L=S_0$). We get that $B=-S_i$ and $A=-C_0^{-1}C_i$. This gives the recursion step to construct $C_i$ and $S_i$ using $C_j,S_j$ with $j<i$. Repeating this recursion procedure, we construct all summands $C_i$ of $C$.

Finally, notice that there exists a formal series $\tilde C$ such that $\tilde C C=C\tilde C=\Id$. The proof is analogous. We denote this series by $C^{-1}$. Theorem is proved.
\end{proof}

\begin{Comment} \label{rem:parC}
\rm{
In the last step of the proof of Theorem \ref{t1}, the matrix $C_i$ is defined up to the addition of an arbitrary diagonal matrix. This implies that the family of such invertible formal differential matrix-valued series $C$ is rather large. However, if one fixes $C_0$ and assumes in addition that all $C_i$ have zeros on the diagonal, then the corresponding $C$ is unique
}    
\end{Comment}

The next statement describes the centralizer of $\mathbb L$.
\begin{Theorem}\label{t2}
Let $\mathbb{L}$ and $C$ be as in Theorem \ref{t1}. Suppose that $A$ is a matrix whose components are differential series such that
$$
[\mathbb{L},A]=0.
$$
Then the matrix $A_{\mathrm{diag}}:=CAC^{-1}$ is diagonal. Moreover, each entry of $A_{\mathrm{diag}}$ is constant 
(that is, an element of $\mathbb{C}$).

\end{Theorem}

\begin{proof}
By Theorem \ref{t1}, there exists $C$ such that $C^{-1}\mathbb{L}C=D+L_{\mathrm{diag}}$. By the definition of the commutator, we have
$$
0=C^{-1}[\mathbb{L},A]C=[C^{-1}\mathbb{L}C,C^{-1}AC]=[D+L_{\mathrm{diag}}, B]\quad\text{where}\quad B=C^{-1}AC.
$$
The entries of the matrices $L_{\mathrm{diag}}$ and $B$ are differential series in the components $L^i_j$ of $L$,
$$
L_{\mathrm{diag}}=S_0+S_1+\dots,\qquad B=B_0+B_1+\dots,
$$
where $S_i$ and $B_i$ are matrices whose entries are homogeneous differential polynomials of degree $i$ in the components of $L$. Moreover, each $S_i$ is diagonal, and $S_0$ has a simple spectrum. We obtain the equality
\begin{equation}\label{eq:com}
0=[D+S_0+S_1+\dots,B_0+B_1+\dots],
\end{equation}
from which we will prove by induction that $B_0$ is diagonal and has constant entries, and that all other $B_i\equiv 0$.

 First, equating the terms of order zero, we obtain $[B_0,S_0]=0$. As $S_0$ is diagonal with a simple spectrum, $B_0$ is diagonal as well, which of course implies $[S_1,B_0]=0$.

Equating the terms of order one, we obtain the equation

\begin{equation}\label{eq:Bx}
0 = \frac{\partial  B_0}{\partial x}+ [S_0, B_1] + [S_1, B_0] = (B_0)_x + [S_0, B_1].
\end{equation}
Let us explain how $\tfrac{\partial B_0}{\partial x}$ appeared in the previous formula. The term $[D,B_0]$, by the definition of the commutator of two operators, acts on $\psi=(\psi^1,\dots,\psi^m)^T$ by
$D(B_0\psi)-B_0D\psi=\tfrac{\partial B_0}{\partial x}\psi$.

Since $S_0$ is diagonal with simple spectrum, $[S_0,B_1]$ has zeros on the diagonal. This implies that $\tfrac{\partial B_0}{\partial x}=0$, so the entries of $B_0$ are constants, and $B_1$ is diagonal. This is the base of the induction.

Now assume that the statement holds for $k-1\geq 1$, that is, assume that $B_0$ is diagonal and has constant entries, that $B_1=B_2=\dots=B_{k-2}=0$, and that $B_{k-1}$ is diagonal. We need to show that $B_{k-1}\equiv 0$ and that $B_k$ is diagonal. Collecting terms of order $k$ in \eqref{eq:com}, we obtain
\begin{equation} \label{eq:question}
0=D(B_{k-1})+[S_0,B_k]+\dots+[S_k,B_0]=D(B_{k-1})+[S_0,B_k].
\end{equation}
Above, $D(B_{k-1})$ means the operation \eqref{eq:D} applied to all components of $B_{k-1}$. It is different from $DB_{k-1}$, which denotes the operator acting on $\psi$ by $D(B_{k-1}\psi)$.

The matrix $[S_0,B_k]$ has zeros on the diagonal, which implies $D(B_{k-1})=0$ and hence $B_{k-1}\equiv 0$. We also obtain $[S_0,B_k]=0$ and thus $B_k$ is diagonal. Theorem is proved.
\end{proof}

\begin{Comment}
\rm{
In the construction in Theorem \ref{t1}, we may assume that the components $L^i_j$ depend analytically on an additional parameter. Recall that  the components of $L$ from Theorem \ref{t2} depend on $\mu$. In view of Remark \ref{rem:parC}, the components of $C$, and therefore also the components of $L_{\mathrm{diag}}$ and the diagonal entries $a_i$ of $A_{\mathrm{diag}}$, will depend analytically on the same parameter.
}    
\end{Comment}

Let us now explain how Theorem \ref{t1} allows us to find conservation laws from the Lax representation. First observe that the Lax equation \eqref{eq:Lax} with $\mathbb{L}=D+L$ and $L\in\operatorname{Mat}(m)$ can be equivalently rewritten as
$
0=[D+L,\partial_t+\mathbb{P}].
$
Here $\partial_t$ is an evolutionary vector field (see Section \ref{sec1}), and the components of $\mathbb{P}$ are differential series.

For $C$ from Theorem \ref{t1} (whose components are differential series), we obtain
$$
\begin{aligned}
0=&C^{-1}[D+L,\partial_t+\mathbb{P}]C=[C^{-1}(D+L)C,C^{-1}(\partial_t+\mathbb{P})C] 
=[D+L_{\mathrm{diag}},C^{-1}\partial_t C + C^{-1}\mathbb{P}C] \\
=&[D+L_{\mathrm{diag}},\partial_t+\underbrace{C^{-1}\partial_t(C)+C^{-1}\mathbb{P}C}_{F}] 
=[D,F]-[\partial_t,L_{\mathrm{diag}}]+[L_{\mathrm{diag}},F].
\end{aligned}
$$
Above, $\partial_t(C)$ is the matrix whose components are differential series obtained by applying the evolutionary vector field $\partial_t$ to the components of $C$. The explanation of why we may replace the term $C^{-1}\partial_t C$ by $\partial_t + C^{-1}(\partial_t C)$ is similar to the explanation of why we replaced $B_0$ by $(B_0)_x$ in \eqref{eq:com} and \eqref{eq:Bx}.

Next, observe that, as $L_{\textrm{diag}}$ is diagonal, the diagonal elements of $[L_{\textrm{diag}},  F] $ are zero, implying that 
the diagonal elements of 
$$[D, F] - [\partial_t,L_{\mathrm{diag}}] = D(F) -\partial_t( L_{\mathrm{diag}})$$
vanish. Hence, the diagonal components of $L_{\mathrm{diag}}$ are conservation laws. Note that in the case of the Lax operator for the BKM systems, since $m=2$ and $L$ is trace-free, the $(1,1)$-entry of $L_{\mathrm{diag}}$ equals minus the $(2,2)$-entry. As the matrix $L$ depends on the parameter $\mu$, the resulting conservation law also depends on $\mu$. In particular, expanding it as a power series in $\mu$, we obtain a hierarchy of conservation laws. We again note that one can make this expansion at any point $\mu_0\in\mathbb{C}\cup{\infty}$; if we take $\mu_0$ to be a root of the polynomial $m$, the resulting conservation laws are differential polynomials.

Let us now explain the Lax-pair geometry behind the constraint in the BKM equations, and the choice of the word ``normalization'' in Theorem \ref{t4}. First, note that the parameterised Lax equations \eqref{eq:laxBKM} and \eqref{eq:laxBKMinf} have an additional factor with denominator $\lambda-\mu$. The assumption that the equation has no pole at $\mu=\lambda$ is then equivalent to
$$
0=[\mathbb{P}(\lambda),\mathbb{L}(\lambda,\lambda)].
$$
In view of Theorem \ref{t5}, this implies that the eigenvalues of $\mathbb{P}(\lambda)$ are constant. Next, since we work in dimension $2$ and $\mathbb{P}$ is trace-free, the latter condition is equivalent to the condition that $\det(\mathbb{P})$ is constant, which gives our constraint. The ``normalization'' is just the choice of a specific value for this constant.


\section{Conclusion}
As explained in the Introduction, BKM systems were constructed in \cite{nijapp4} within the Nijenhuis geometry research programme initiated in \cite{BMMT,nij1}; the construction was a side result of the study of compatible $\infty$-dimensional geometric Poisson brackets.
A natural question is whether one can combine methods developed for and within the theory of integrable systems with those coming from Nijenhuis geometry.
Our recent publications \cite{nijapp5,nij4} show that the link between Nijenhuis geometry and the theory of integrable systems is useful for Nijenhuis geometry, as objects that appeared within integrable systems theory, such as conservation laws and symmetries, are very helpful in Nijenhuis geometry, and in particular were applied in the classical theory of geodesically equivalent metrics \cite{nijapp5,Fomenko_80}.

It is natural to  go  also in the  other direction, namely to  study the  methods and objects that have appeared and are used in the theory of integrable systems  from the viewpoint of Nijenhuis geometry.
As explained and recalled in the Introduction, many integrable systems of applied and mathematical interest are special cases of BKM systems.
Translating the standard methods of the theory of integrable systems to the language of Nijenhuis geometry may provide a universal treatment of many known and newly obtained integrable systems. Our recent paper \cite{BKMreduction}, see also the survey \cite{BKMmatrix}, is a successful demonstration of this approach: in that paper
we constructed finite-dimensional reductions of BKM systems, recovering, as special cases, known finite-dimensional reductions of certain systems previously studied in the context of mathematical physics, e.g., in \cite{moser, BM2006,BS2023, Dubrovin1975}.

The goal of this paper was to understand whether the Lax pair tool can be generalised to all BKM systems, and we have indeed shown in Theorem \ref{t5}  and Section \ref{sec5} that this is the case.

The Lax pairs we constructed come in the form of a family. That is, besides the ``spectral'' parameter $\mu$, which comes only in the $\mathbb{L}$-element,  there is a dependence  of $\mathbb{P}$  on $\lambda$, which parametrizes the BKM equation. This implies that
the Lax pairs generate those for the components of the hierarchy, and therefore the  symmetries of an initial system,   in a manner similar to the way hierarchies are generated from the initial families of BKM equations. In Section \ref{sec3} we explained how the Lax pair  generates the hierarchy of the conservation laws. 

Let us now discuss the next natural steps and possible applications of our results.
Recall that in the KdV case and certain other cases, the existence of Lax pairs allows the so-called linearization via a compatibility condition; let us recall what this means.
Introduce a spectral problem for an eigenfunction $\psi$:
\begin{equation}\label{eq:spec} \mathbb{L}\psi=\mu\psi, \qquad \psi_t=\mathbb{P}\psi. \end{equation} The compatibility of these two linear equations implies \[ (\mathbb{L}\psi)_t = \mathbb{L}_t\psi + \mathbb{L}\psi_t = \mu \psi_t \;\;\Longleftrightarrow\;\; \mathbb{L}_t=[\mathbb{P},\mathbb{L}], \]
recovering \eqref{eq:Lax}. Eliminating $\psi$ typically yields the original nonlinear PDE/ODE for the coefficients in $\mathbb{L}$.

For the  BKM I and BKM II systems,  the associated linear problem on the line is
\begin{equation} \label{eq:1}
\psi_{xx} + \frac{1}{2}\frac{\sigma(\mu)}{m(\mu)} \psi = 0,
\end{equation}
and for the BKM III and BKM IV systems the the associated linear problem is the barred version of \eqref{eq:1}. 
Recall that the function $m(\mu)$  (and $\bar m(\mu)$) is a polynomial of degree $n$ in $\mu$, while $\sigma(\mu)$ depends rationally on $\mu$, with coefficients that are functions of the field variables. The KdV and Camassa–Holm equations
are special cases of BKM systems with $n=1$ and $m(\mu)\equiv 1$. In this case, $\sigma(\mu)$ is linear in $\mu$ (with constant leading coefficient), so \eqref{eq:1} is exactly the spectral problem for the Sturm–Liouville operator, well known and well studied both for rapidly decaying and quasi-periodic functions. The equation \eqref{eq:1} with $m \equiv 1$ and monic polynomial $\sigma(\mu)$ of arbitrary degree $n$ was studied by 
 L.~Mart\'{\i}nez Alonso in \cite{alonso} and later in, e.g., \cite{af, nab, HrMa}; the last two references contain, in particular, an overview of the literature and new results concerning the spectral properties of such operators, as well as the related inverse scattering problem. Many of these papers focus on the $n=2$ case because of its relation to the Kaup–Boussinesq and Ito systems.
The time evolution,  for the vector-function  $\Psi= (\psi, \psi_x)^T$,   is given by
\[
\Psi_t = \dfrac{1}{\mu-\lambda} \mathbb{P}(\mu,\lambda)\Psi
\quad\Longleftrightarrow\quad
\begin{cases}
\displaystyle
\psi_{t}
= \dfrac{1}{\mu-\lambda}
\Bigl(-\tfrac12 q_x(\lambda)\,\psi + q(\lambda)\,\psi_x\Bigr),\\[0.8em]
\displaystyle
\psi_{tx}
= \dfrac{1}{\mu-\lambda}
\Bigl(\bigl(-\tfrac12 q_{xx}(\lambda) - \tfrac12 \tfrac{\sigma(\mu)}{m(\mu)} q(\lambda)\bigr)\psi 
+ \tfrac12 q_x(\lambda)\,\psi_x\Bigr).
\end{cases}
\]
The compatibility condition
$ 
\Psi_{xt} = \Psi_{tx}
$ 
 is equivalent to the parameterized Lax equation \eqref{eq:laxBKM}, or to the BKM system in the form \eqref{eq:parametricform1} from  Theorem \ref{t5}.

Note that, for  BKM systems, for the associated linear problem the potential \eqref{eq:1}
on the line depends rationally on the parameter and may include both stationary (when the polynomial $m $ is not a constant)   and moving (when $\sigma$ has a nontrivial denominator) poles. It seems that in this setting the problem has not been widely studied, even when $\sigma(\mu)$ is a polynomial (i.e., if $r=0$) and $m(\mu)\neq 1$. We did not find the most general case of \eqref{eq:1}, with $\sigma$ depending rationally on $\mu$, in the literature.

The natural next step would be to study the corresponding linear problem, with the long-term goal of mimicking
the inverse scattering method, which has been successfully applied to certain previously known and studied special cases of BKM systems, for general BKM systems as well. The paper \cite{km1}, in which we discussed
certain well-known phenomena of the linear problem associated with the KdV system from the viewpoint of BKM systems, gives first results in this direction. We cordially invite our colleagues to join the investigation.

\subsection*{Acknowledgments.}    A.\,K. was supported by the Ministry of Science and Higher Education of the Republic of Kazakhstan (grant No. AP23483476), and   V.\,M. by  the DFG (projects 455806247 and 529233771) and the ARC Discovery Programme DP210100951. We thank Alexey Bolsinov for useful discussions during the work on the paper and for careful reading of the final version.

  Data sharing is not applicable to this article, as no datasets were generated or analysed during the current study. The authors declare no conflicts of interest.



\printbibliography
\end{document}